\documentclass[10pt,journal,compsoc]{IEEEtran}
\usepackage{amsmath,amssymb,amsfonts}
\usepackage{algorithm}
\usepackage{algorithmic}
\usepackage{orcidlink}
\usepackage{balance}
\allowdisplaybreaks[4]
\usepackage{graphicx}
\usepackage{booktabs}
\usepackage{colortbl}
\usepackage{makecell}
\usepackage{multirow} 
\usepackage{subfig}
\usepackage{textcomp}
\usepackage{xcolor}
\usepackage{enumerate} 
\usepackage{enumitem}
\usepackage{pdfpages}
\usepackage{hyperref}
\usepackage{cite}

\newtheorem{ass}{Assumption}
\newtheorem{pro}{Proposition}
\newtheorem{lem}{Lemma}
\newtheorem{rem}{Remark}

\newtheorem{thm}{Theorem}

\newtheorem{definition}{Definition}
\hypersetup{colorlinks=true,linkcolor=black,citecolor=black,urlcolor=black}
\newenvironment{proof}{{\it Proof}.\ }{\hfill $\blacksquare$\par}
\newcommand{\refappendix}[1]{\hyperref[#1]{Appendix~\ref*{#1}}}

\begin{document}

\title{Adaptive Federated Learning via New Entropy Approach}

\author{Shensheng~Zheng$^{\#}$\orcidlink{0000-0002-2951-3737}, Wenhao~Yuan$^{\#}$\orcidlink{0009-0001-6625-7496},~\IEEEmembership{Student Member, IEEE}, Xuehe~Wang\orcidlink{0000-0002-6910-468X},~\IEEEmembership{Member, IEEE}, Lingjie~Duan\orcidlink{0000-0002-0217-6507},~\IEEEmembership{Senior Member, IEEE}

\IEEEcompsocitemizethanks{\IEEEcompsocthanksitem Shensheng Zheng, Wenhao Yuan, and Xuehe Wang are with the School of Artificial Intelligence, Sun Yat-sen University, Zhuhai 519082, China. Email: \href{mailto:zhengshsh7@mail2.sysu.edu.cn}{zhengshsh7@mail2.sysu.edu.cn}, \href{mailto:yuanwh7@mail2.sysu.edu.cn}{yuanwh7@mail2.sysu.edu.cn}, \href{mailto:wangxuehe@mail.sysu.edu.cn}{wangxuehe@mail.sysu.edu.cn}.
\IEEEcompsocthanksitem Lingjie Duan is with the Pillar of Engineering Systems and Design, Singapore University of Technology and Design, 467372, Singapore. Email:  \href{mailto:lingjie_duan@sutd.edu.sg}{lingjie\_duan@sutd.edu.sg}.
}
\thanks{$\#$ Shensheng Zheng and Wenhao Yuan have contributed to this work equally.\\
This work was supported by the National Natural Science Foundation of China (NSFC) under Grant No. 62206320. (Corresponding author: Xuehe Wang.)}}

\IEEEtitleabstractindextext{%
\begin{abstract}
Federated Learning (FL) has emerged as a prominent distributed machine learning framework that enables geographically discrete clients to train a global model collaboratively while preserving their privacy-sensitive data. However, due to the non-independent-and-identically-distributed (Non-IID) data generated by heterogeneous clients, the performances of the conventional federated optimization schemes such as FedAvg and its variants deteriorate, requiring the design to adaptively adjust specific model parameters to alleviate the negative influence of heterogeneity. In this paper, by leveraging entropy as a new metric for assessing the degree of system disorder, we propose an adaptive FEDerated learning algorithm based on ENTropy theory (FedEnt) to alleviate the parameter deviation among heterogeneous clients and achieve fast convergence. Nevertheless, given the data disparity and parameter deviation of heterogeneous clients, determining the optimal dynamic learning rate for each client becomes a challenging task as there is no communication among participating clients during the local training epochs. To enable a decentralized learning rate for each participating client, we first introduce the mean-field terms to estimate the components associated with other clients' local parameters. Furthermore, we provide rigorous theoretical analysis on the existence and determination of the mean-field estimators. Based on the mean-field estimators, the closed-form adaptive learning rate for each client is derived by constructing the Hamilton equation. Moreover, the convergence rate of our proposed FedEnt is proved. The extensive experimental results on the real-world datasets (i.e., MNIST, EMNIST-L, CIFAR10, and CIFAR100) show that our FedEnt algorithm surpasses FedAvg and its variants (i.e., FedAdam, FedProx, and FedDyn) under Non-IID settings and achieves a faster convergence rate.  
\end{abstract}

\begin{IEEEkeywords}
Adaptive Federated Optimization, Entropy Theory, Mean-Field Analysis.
\end{IEEEkeywords}}

\maketitle

\IEEEdisplaynontitleabstractindextext

\IEEEpeerreviewmaketitle

\IEEEraisesectionheading{\section{Introduction}\label{introduction}}
\IEEEPARstart{T}he rapid growth of the Internet of Things (IoT) and Artificial Intelligence (AI) has facilitated the utilization of large volumes of discrete client data to train efficient machine learning (ML) models, particularly for applications such as recommendation systems \cite{ghosh2020efficient}, healthcare informatics \cite{xu2021federated}, edge computing \cite{xia2021survey}, etc. However, traditional ML approaches require the collection of massive data, which raises concerns regarding personal data privacy \cite{bonawitz2017practical}. To address such issue while economizing computational and communication resources while safeguarding users' data privacy, federated learning has emerged as a viable solution, which enables resource-constrained discrete clients to jointly train a global model under the orchestration of a central server while keeping privacy-sensitive data localized \cite{yang2019federated}.

The widely adopted FL algorithms, such as Federated Averaging (FedAvg) and its variants \cite{r1}, typically consist of a central server and a group of distributed clients. In each global training iteration, the discrete clients first receive the global model parameter from the central server, followed by utilizing their local data to derive the respective local model parameters. Subsequently, the central server aggregates the local model parameters contributed by each participating client and updates the global model. The training process proceeds until the training loss of the global model falls below a predefined threshold \cite{r1}. This framework efficiently leverages the computational consumption and storage resources of each client, all the while maintaining the security of privacy-sensitive data \cite{tran2019federated}.

Although FedAvg performs well on Gboard application \cite{hard2018federated}, recent literature \cite{zhao2018federated, hsu2019measuring, karimireddy2020scaffold} has highlighted convergence issues in specific scenarios. \cite{sattler2019robust}  observes that FedAvg exhibits a slow convergence rate and deteriorated accuracy when applied to non-independent-and-identically-distributed (Non-IID) datasets, which is caused by client parameter drifting and lack of adaptivity to the different models \cite{reddi2020adaptive}. Besides,  another unavoidable issue in FL pertains to the substantial communication overhead incurred by the dissemination of model updates. Adhering to the previously delineated protocol necessitates that each participating client transmit a comprehensive model update at every iteration of training, which will introduce a huge communication burden \cite{sattler2019robust}. Although clients can increase the number of local training epochs, the inability of participating clients to directly communicate during the local update process may result in model parameter deviation, thereby reducing the convergence speed and accuracy of the model \cite{reddi2020adaptive}.


Substantial federated optimization methods have been proposed to address the convergence challenges associated with Non-IID datasets, i.e.,  partial federated optimization, asynchronous federated optimization, and adaptive federated optimization. For partial federated optimization, some works adopt a client selection mechanism to choose geographically discrete clients with high-quality local information, such as computation and communication resources, gradient information, and local dataset \cite{fu2023client}. Asynchronous federated optimization performs global aggregation as soon as the central server collects a single local update from any arbitrary mobile client node, while this practice may potentially result in significant deterioration on the global model due to the influence of a single client's local model, thereby leading to suboptimal convergence performance \cite{zhang2023timelyfl}. Adaptive federated optimization focuses on adaptively adjusting the model hyperparameter (e.g., aggregation frequency \cite{liu2021fedpa, liu2020client, sun2020adaptive}, aggregation weight \cite{wu2021fast}, the batchsize of local training \cite{ren2020accelerating, zhang2021adaptive}, and the allocated bandwidth \cite{xu2021adaptive}) to achieve better global model performance.

However, there are three issues in the above FL optimization methods. Firstly, most methods \cite{wang2019adaptive, li2020federated, nguyen2020fast, chen2019communication, wu2021fast, nishio2019client} neglect the optimization on learning rate and tend to assign a static learning rate during model training, which is suboptimal as variations in mobile devices and clients' datasets can lead to disparities in local parameters and a slow convergence rate. Secondly, the convergence analysis of adaptive FL algorithms in \cite{reddi2020adaptive} relies on multiple hyperparameters and fails to provide precise adaptation for each training iteration. Last, due to the lack of communication among clients during local training, how to design the decentralized adaptive algorithm that accelerates the convergence rate of the global model is also a problem. In this paper, by utilizing a new entropy term to measure the diversity among all clients' local model parameters, we adopt an adaptive learning rate scheme for each client to achieve fast convergence. Our key novelty and main contributions are summarized as follows.

\begin{itemize} [itemsep=0pt, leftmargin=*, align=right]
\item\emph{New Entropy-based adaptive federated learning:} 
To our best knowledge, this is the first paper to utilize entropy theory to optimize the decentralized adaptive learning rate for each client under the Non-IID data distribution. To alleviate the negative influence of heterogeneity, an entropy term that measures the diversity among the local model parameters of all clients is taken into consideration for each client's model updates. By employing a mean-field scheme to continuously estimate other clients’ local information in the entropy term, we design a decentralized adaptive learning rate for each client to achieve fast convergence, without requiring the clients to exchange their local parameters with each other.

\item\emph{Decentralized learning rate design via mean-field analysis:}  
Due to the data disparity of heterogeneous clients, the learning rate of individual clients should be dynamically adjusted based on other clients' local information to avoid parameter deviation and achieve faster convergence. Given that there is no communication during local training epochs, a mean-field scheme is introduced to estimate the terms relating to other clients' local parameters over time, and thus the decentralized adaptive learning rate for each client is obtained with closed-form by constructing the Hamilton equation. Moreover, the existence of the mean-field estimators is demonstrated, and an iterative algorithm is further proposed for their determination. 

\item\emph{Convergence analysis for the proposed FedEnt:} 
We theoretically analyze the global model convergence property and provide the convergence upper bound and convergence rate for our proposed adaptive FL algorithm FedEnt. Extensive experimental results on real-world datasets under Non-IID distribution demonstrate that FedEnt achieves better model performance and effectively alleviates the deviation among heterogeneous clients compared with state-of-the-art FL algorithms (i.e., FedAvg, FedAdam, FedProx, and FedDyn).
\end{itemize}

The rest of this paper is organized as follows. We review the related work in Section \ref{related_wrok}. The system model and problem formulation are presented in Section \ref{sec_systemmodel}. In Section \ref{subsec_adlr}, we design the decentralized adaptive learning rate via the entropy approach. Then, the convergence analysis is provided in Section \ref{sec_conv_an}. Experimental results are shown in Section \ref{sec_experiments}. Finally, we conclude this paper in Section \ref{sec_conclusion}.

\section{Related Work} \label{related_wrok}
In this section, we briefly review the related literature on partial optimization, asynchronous optimization, and adaptive optimization in federated learning. 
 
\subsection{Partial Optimization in FL}
Cho \emph{et al.} \cite{cho2022towards} propose an FL framework with biased client selection that is cognizant of the training progress of each participating client. PyramidFL \cite{li2022pyramidfl} focuses on the difference between selected and unselected participants, and fully exploits the system and data heterogeneity within selected clients to efficiently profile their utility. Nguyen \emph{et al.} \cite{nguyen2020fast} predominantly rely on gradient information from clients and exclude ones with negative inner products between their gradient and the global gradient vector from the FL training. Huang \emph{et al.} \cite{huang2023maverick} propose a strategy for client selection based on the Wasserstein distance between local and global data distributions and establish a convergence bound for it. Nishio \emph{et al.} \cite{nishio2019client} present the FedCS algorithm, which selects nodes based on the resource conditions of local clients rather than employing random selection. Liu \emph{et al.} \cite{liufeddwa} propose FedDWA, aiming to minimize the gap between personalized models and reference models to customize aggregation weights for each client. In addition, the literatures \cite{zhang2021optimizing, yang2023detfed, zhang2021deep} focus on the application of FL in Industrial IoT, constructs new objective optimization functions based on efficient resource utilization, and utilizes deep reinforcement learning for optimal resource scheduling decisions.

\begin{figure*}[t]
\centerline{\includegraphics[width=1.0\textwidth, trim=10 10 10 10,clip]{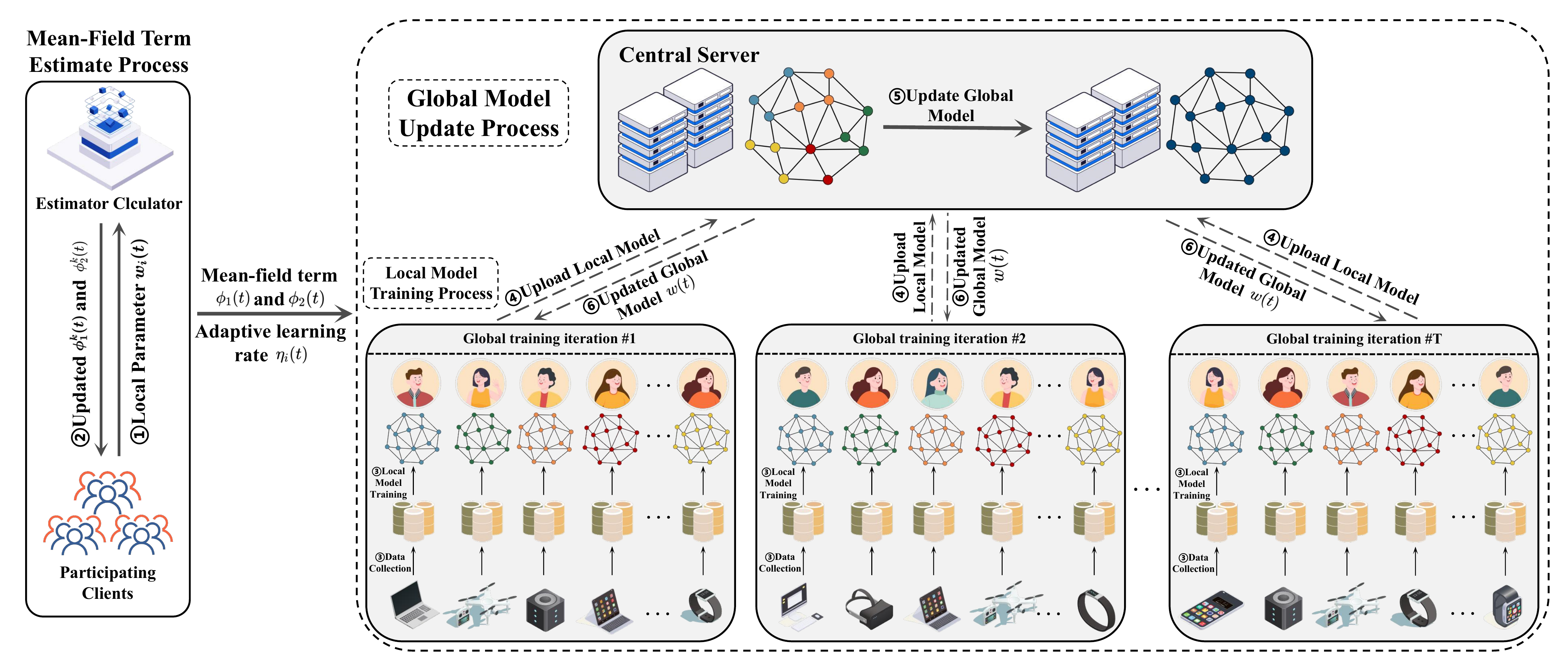}}
\caption{The Federated Learning Framework with Entropy Approach. }
\label{framework}
\vspace{-15pt}
\end{figure*}

\subsection{Asynchronous Optimization in FL}
Xie \emph{et al.} \cite{xie2019asynchronous} employ a weighted averaging approach for the asynchronous updating of the global model, with the mixing weight being dynamically determined as a function of staleness. Ma \emph{et al.} \cite{ma2021fedsa} introduce a semi-asynchronous federated learning mechanism known as FedSA, in which the parameter server aggregates a specified number of local models based on their arrival order during each global iteration. Xu \emph{et al.} \cite{xu2023fedlc} present a novel asynchronous FL mechanism named FedLC, allowing each heterogeneous device to transmit gradients to other devices with varying data distributions for facilitating local collaboration and enhancing model generality. Zhang \emph{et al.} propose TimelyFL \cite{zhang2023timelyfl}, a heterogeneity-aware asynchronous federated learning framework that incorporates adaptive partial training. Nguyen \emph{et al.} \cite{nguyen2022federated} introduce FedBuff, an innovative asynchronous federated optimization framework featuring buffered asynchronous aggregation to achieve scalability and privacy against the honest-but-curious threat model, incorporating secure aggregation and differential privacy.

\begin{table}[t]
\setlength{\abovecaptionskip}{0cm} 
\caption{Key Notations and Corresponding Meanings}
\begin{center}
\renewcommand\arraystretch{1.2}
\begin{tabular}{p{0.5in}p{2.65in}}
\toprule[0.8pt]
\makecell[l]{\textbf{Symbol}}&\makecell[c]{\textbf{Description}}\\
\midrule[0.8pt]
$N$ & number of participating clients   \\

$T$ & total number of global training iteration  \\

$\mathcal{D}_{i}$ & private data on client $i$   \\

$D_{i}$ & the datasize of $\mathcal{D}_{i}$   \\

$\boldsymbol{w_{i}}(t)$ & local model parameter of client $i$ at global iteration $t$    \\

$\boldsymbol{w}(t)$ & global model parameter at global iteration $t$   \\

$F_{i}(\boldsymbol{w}(t))$ & local loss function of client $i$  \\

$F(\boldsymbol{w})$ & global loss function  \\

$\nabla F$ & gradient of loss function  \\

$\eta_{i}(t)$ & learning rate of client $i$ at the $t$-th global iteration  \\ 

$\boldsymbol{\phi_{1}}(t)$ & mean-field terms for estimating the global weight \\

$\phi_{2}(t)$ & mean-field terms for estimating the inner product of other clients' local parameter \\

\bottomrule[0.8pt]
\end{tabular}
\label{tab1}
\end{center}
\vspace{-15pt}
\end{table}

\subsection{Adaptive Optimization in FL}
Wang \emph{et al.} \cite{wang2019adaptive} adjust the frequency of local updates in real-time dynamically between consecutive global aggregations, with the objective of minimizing learning loss in the computational system when the total resource allocation is fixed. Reddi \emph{et al.} \cite{reddi2020adaptive} employ a pseudo-gradient difference method to adaptively update gradients. Wu \emph{et al.} \cite{wu2021fast} adjust client weights by quantifying their contributions through the angle between the local and the global gradient vector. Chen \emph{et al.} \cite{chen2019communication} distribute different update weights based on the number of local communication epochs. Ma \emph{et al.} \cite{ma2021fedsa} compute the dynamic learning rate based on aggregation frequency while neglecting to account for variations among different local parameters. Md \emph{et al.} \cite{uddin2023arfl} propose an adaptive model update strategy that optimally allocates local epochs and dynamically adjusts the learning rate while regularizing the conventional objective function by adding a proximal term. He \emph{et al.} \cite{he2023clustered} propose an adaptive layer-wise weight perturbation technique by introducing fine-grained noise to enhance model accuracy.

Note that most aforementioned literature on adaptive federated learning overlooks the optimization of the learning rate to alleviate the negative influence of client heterogeneity, as well as the associated convergence property analyses. In contrast to the previous work, we derive the closed-form solution to the adaptive learning rate for each client in each global iteration by utilizing entropy and mean-field schemes. 

\section{System Model and Problem Formulation}\label{sec_systemmodel}

In this section, we first introduce the preliminaries of the FL framework in Section \ref{subsec_FL_model}, and then formulate our adaptive FL optimization problem by considering the deviation among the local parameters at each global iteration in Section \ref{problem_formulation}. Our proposed federated learning framework with an entropy approach is shown in Fig. \ref{framework}.

\subsection{Standard Federated Learning Model}\label{subsec_FL_model}

Federated learning is a distributed ML paradigm in which a large number of mobile clients coordinate with the central server to train a shared model without sharing their private data \cite{r1}. In each global iteration, the participating client utilizes their local data to perform several epochs of mini-batch stochastic gradient descent (SGD) to update the local model, and the central server aggregates the local parameters of all clients to update the global model for the next iteration.

Given $N$ clients participating in the FL model training process, we assume that each client $i \! \in \! S \! = \! \{1,2,\ldots, N\}$ leverages its local private dataset $\mathcal{D}_i$ with data size $D_i$. We denote the collection of data samples in $\mathcal{D}_i$ by $\{\boldsymbol{x_i^{j}}, y_i^{j}\}_{j=1}^{D_i}$, where $\boldsymbol{x_i^{j}} \in \mathbb{R}^{d}$ and $y_i^{j} \in \mathbb{R}$ refer to the input sample vector and the labeled output value, respectively. The loss function $f_i^{j}(\boldsymbol{w})$ is used to measure the divergence between the real output $y_i^{j}$ and predicted output ${\hat{y}}_i^{j}$. Thus, the loss function $F_{i}(\boldsymbol{w})$ of client $i$ on the dataset $\mathcal{D}_i$ is defined as follows:
\begin{align}\label{client_own_loss}
\textstyle F_i(\boldsymbol{w}) = \frac{1}{D_i}\sum_{j \in {\mathcal{D}}_i} f_i^{j}(\boldsymbol{w}).
\end{align}

Without loss of generality, the total training iteration conducted by the federated learning process is denoted by $T$ and we adopt a general principle of SGD optimization to update the local model \cite{wu2021fast}. At iteration $t + 1 \in \{1, 2, \ldots, T\}$, each client uses its own dataset and the learning rate to update its local parameters based on the last global parameter $\boldsymbol{w}(t)$ received from the central server:
\begin{align}\label{w_static_update}
\boldsymbol{w_i}(t+1) = \boldsymbol{w}(t) - \eta_i\nabla F_i(\boldsymbol{w}(t)),
\end{align}
where $\eta_i$ is the learning rate of client $i$ and $\nabla F_i(\boldsymbol{w}(t))$ is $i$-th client's gradient at global iteration $t$. The central server aggregates the local parameters uploaded by $N$ participating clients to update the new global parameter:
\begin{align}\label{w_global_static_update}
\boldsymbol{w}(t+1)= \textstyle \sum_{i=1}^{N}\theta_{i} \boldsymbol{w_i}(t+1),
\end{align}
where $\theta_{i} \! = \! \frac{D_i}{\sum_{j=1}^{N}D_j}$ represents the aggregation weight of client $i$. Then, the central server dispatches the updated global parameter $\boldsymbol{w}(t+1)$ to all participating clients for the fresh iteration of global training. Finally, the objective of the FL model is to find the optimal global parameter $\boldsymbol{w}$ to minimize the global loss function $F(\boldsymbol{w})$, i.e.,
\begin{align}
\min_{\boldsymbol{w}} F(\boldsymbol{w}) = \textstyle \sum_{i=1}^{N} \theta_{i} F_{i}(\boldsymbol{w})
\end{align}

For the subsequent theoretical analysis, we adopt the following assumptions and lemma in our demonstration, which are widely used in federated optimization problems \cite{r1, nguyen2020fast, ward2020adagrad}.
\begin{ass}\label{Bounded_Gradients}
(Bounded Gradient) The local gradient function $\nabla F_{i}(\boldsymbol{w})$ have $D$-bound gradients, i.e., for any $i \! \in \! \{1, \\ 2, \ldots, N\}$, we have $\|\nabla F_{i}(\boldsymbol{w}) \| \leq D$.
\end{ass}
\begin{ass}\label{Lipschitz_Smoothness}
($L$-Lipschitz Smoothness)  The local loss function $F_{i}(\boldsymbol{w})$ is $L$-Lipschitz smoothness for each participating clients $i \in \{1,2, \ldots, N\}$, i.e, 
\begin{align}
\| \nabla F_{i}(\boldsymbol{w}) - \nabla F_{i}(\boldsymbol{w'})  \| \leq L \| \boldsymbol{w} - \boldsymbol{w'} \|,
\end{align}
where $\boldsymbol{w}$ and $\boldsymbol{w'}$ are any two model parameters.
\end{ass}

\begin{lem} \label{PL_inequality}
[Polyak-\L{}ojasiewicz inequality] Denote the optimal global parameter as $\boldsymbol{w}^{*}$, the $L$-Lipschitz continuous global loss function $F(\boldsymbol{w})$ satisfies Polyak-\L{}ojasiewicz condition, i.e., the following holds for any $\delta > 0$ and $\boldsymbol{w} \in \mathbb{R}^{d}$: 
\begin{align}
\|\nabla{F(\boldsymbol{w})}\|^{2} \geq 2 \delta(F(\boldsymbol{w})-F(\boldsymbol{w}^{*})). 
\end{align}
\end{lem}

\begin{rem}
In practice, the value of $L$ can be estimated or selected through experimentation or theoretical analysis. If the expression of the loss function is known, $L$ can be directly calculated through analytical methods. For instance, for known loss functions such as quadratic functions or functions with bounded Hessian matrices, $L$ can be directly calculated. Similar to $L$, the value of $D$ can be estimated through experience or by analyzing the properties of the function. If the gradient of the function can be directly calculated, the maximum value of the gradient can be set as the estimator for $D$. As shown in Table \ref{table_estimate}, we conduct numerical experiments to determine the value of $D$ and $L$ under different datasets and local model setup, which also guarantees the global model convergence.
\end{rem}

\begin{table}[t]
\setlength{\abovecaptionskip}{-4pt} 
\caption{Statistical estimates of $D$ and $L$ on FL system}
\renewcommand\arraystretch{1.0}
\begin{center}
\resizebox{0.47\textwidth}{!}{\begin{tabular}{c|c|c|}
\toprule[1pt]
\parbox{5.0cm}{\centering\textbf{Datasets+Local Model}}&\parbox{1cm}{\centering\textbf{$D$} }&\parbox{1cm}{\centering\textbf{$L$} } \\ \cmidrule[0.5pt](l{1pt}r{0pt}){1-3}

\parbox{4.5cm}{\centering\textbf{MNIST+MNIST-CNN}} & 40 & 39011 \\ \cmidrule[0.5pt](l{1pt}r{0pt}){1-3}

\parbox{4.5cm}{\centering\textbf{EMNIST-L+LeNet-5}} & 6770 & 660446 \\ \cmidrule[0.5pt](l{1pt}r{0pt}){1-3}

\parbox{4.5cm}{\centering\textbf{CIFAR10+CIFAR10-CNN}}& 57 & 96767 \\ \cmidrule[0.5pt](l{1pt}r{0pt}){1-3}

\parbox{4.5cm}{\centering\textbf{CIFAR100+VGG-11}}& 79 & 17329 \\ 
\bottomrule[1pt]
\end{tabular}}
\label{table_estimate}
\end{center}
\vspace{-18pt}
\end{table}


\subsection{Problem Formulation for New Entropy-based adaptive Federated Learning} \label{problem_formulation} 
In the conventional FL algorithms, each participating client utilizes a static learning rate to train the local model. However, there exist several drawbacks to such an update rule. First of all, it fails to adapt to the parameter updates for different clients, leading to the local model deviating from the optimal global model when the data distribution and devices of the clients are heterogeneous \cite{hsu2019measuring}. Moreover, the static learning rate strategy mitigates the convergence speed of the global model in comparison to the dynamic learning rate scheme \cite{agrawal2021genetic}. To deal with the above problems, we introduce the following dynamic learning rate approach to update the local model parameters:
\begin{align}\label{dynamic_w_update}
\boldsymbol{w_i}(t+1) = \boldsymbol{w}(t) - \eta_i(t)\nabla F_i(\boldsymbol{w}(t)),
\end{align}
where $\eta_i(t) \! \in \! (0, 1) $ is the learning rate of client $i$ at the global iteration $t$. Recent literature has increasingly focused on entropy-based criteria within adaptive systems \cite{silva2005neural}, highlighting the significance of entropy in the optimization of learning rates in machine learning research  \cite{silva2005neural, santos2004optimization, zhang2016critical}. \cite{ahmed2019understanding} examines how the optimal learning rate correlates with the local curvature of the loss function, observing that higher entropy accelerates this relationship. By integrating entropy, the approach moderates the curvature's fluctuations along the optimization path, thereby enhancing the efficacy of learning rate adjustments. Despite the advantages of using entropy, existing literature seldom considers adjusting the learning rate and controlling the local model parameters based on entropy. Notably, the convergence speed of the global model is directly influenced by the degree of divergence among the local parameters of participating clients, with lesser divergence facilitating quicker convergence \cite{sattler2020clustered}. In other words, the model will tend to converge when the local parameters of each client are similar. Thus, we introduce the following entropy term to measure the divergence among heterogeneous clients in federated learning.
\begin{definition}\label{def1}
To evaluate the degree of system disorder, we utilize the entropy term $\sum_{i=1}^{N}p_{i}(t) \log p_{i}(t)$ to quantify the differences among the local parameters, where  
\begin{align}\label{p_init}
p_{i}(t) = \frac{\theta_{i} \boldsymbol{w_{i}}(t)^{T}\boldsymbol{w_{i}}(t)}{\sum_{j=1}^{N}  \theta_{j} \boldsymbol{w_{j}}(t)^{T}\boldsymbol{w_{j}}(t)} \in [0, 1].
\end{align}
\end{definition}

Accordingly, to achieve a fast FL model convergence, each participating client $i \!\in\! \{1,2,\ldots, N\}$ aims to design the adaptive learning rate $\eta_{i}(t), t \!\in\! \{0,1,  \ldots,  T\}$ that minimizes the divergence among the local parameters, by considering the entropy term $\sum_{i=1}^{N}p_{i}(t) \log p_{i}(t)$ that quantify the differences among the local parameters. In addition, to avoid significant fluctuation during the process of local model training, the size of the learning rate should also be taken into consideration. Thus, the objective function is formulated as follows:
\begin{align}\label{objective}
U_{i}(T) \!= & \footnotesize \min_{\substack{\eta_{i}(t)}} \sum_{t=0}^{T}\! \left(\! \beta \! \sum_{j=1}^{N}p_{j}(t) \log (p_{j}(t)) \!+\!  (1 \!-\! \beta)\eta_{i}^{2}(t)\!\right), \\ 
& s.t.\quad \boldsymbol{w_i}(t+1) = \boldsymbol{w}(t)-\eta_i(t)\nabla F_i(\boldsymbol{w}(t)), \nonumber
\end{align}
where aggregation weight $\beta \!\in\! (0,1)$ is predetermined to measure the degree of system disorder caused by the local model parameter deviation. A high value of $\beta$ demonstrates a high influence of the local parameter deviation on the FL system. 

However, there are two challenges in the above optimization problem with the objective function $U_{i}(T)$ in Eq. (\ref{objective}). Firstly, it is difficult to solve the optimal dynamic learning rate by considering the huge number of learning rate combinations over time. Moreover, the design of $i$-th client's learning rate $\eta_i(t)$ is not only affected by its own local parameters $\boldsymbol{w_{i}}(t)$ but also affected by other clients' local parameters according to the term $\sum_{j=1}^{N}  \theta_{j} \boldsymbol{w_{j}}(t)^{T}\boldsymbol{w_{j}}(t)$, $j \! \in \! S \backslash \{i\}$, which makes the multi-agent joint learning rate design more challenging. To solve the above problem, we first derive the decentralized dynamic learning rate $\eta_i(t)$ for each client $i$ via mean-field approach, then we prove the existence of the approximate mean-field terms via fixed point theorem and propose an iterative algorithm to finalize the learning rate design.

\section{Analysis of Adaptive Learning Rate}\label{subsec_adlr}

In this section, we derive the optimal adaptive learning rate $\eta_{i}(t)$ for each client $i$ through the mean-field approach. We first solve the adaptive learning rate $\eta_{i}(t)$ by constructing the mean-field terms $\boldsymbol{\phi_1}(t)$ and $\phi_2(t)$ to estimate the global parameter $\boldsymbol{w}(t)$ and  $\sum_{j=1}^{N}  \theta_{i} \boldsymbol{w_{i}}(t)^{T}\boldsymbol{w_{i}}(t)$, respectively. Then, we demonstrate the existence of mean-field terms $\boldsymbol{\phi_1}(t)$ and $\phi_2(t)$. Finally, we propose an algorithm to iteratively calculate the adaptive learning rate. 

\subsection{Framework Overview of FedEnt}
The training framework of our proposed FedEnt is shown in Fig.~\ref{framework} and the following overview demonstrates the detailed content of each step:
\begin{itemize}[itemsep=0pt, leftmargin=*, align=right]
\item \textbf{Step 1 (Estimators Initialization):} At the commencement of the fixed point calculation for mean-field terms, each client updates an arbitrary initial local parameter $\boldsymbol{w_i(0)}$ to the estimator calculator. Besides, the estimator calculator will define the fixed point termination thresholds $\varepsilon_{1}$ and $\varepsilon_{2}$ for the following iterative calculation process. 

\item \textbf{Step 2 (Update Estimators and Calculate Fixed Point):} After receiving the local parameter $\boldsymbol{w_i}(t)$ from all clients, the estimator calculator aggregates the local model and updates the mean-field terms $\boldsymbol{\phi_{1}^{k}}(t)$ and $\phi_{2}^{k}(t)$. The iteratively calculating process terminates until the mean-field terms $\boldsymbol{\phi_{1}^{k}}(t) \!-\! \boldsymbol{\phi_{1}^{k-1}}(t) \! < \! \varepsilon_{1}$ and $\phi_{2}^{k}(t) \!-\! \phi_{2}^{k-1}(t) \! < \! \varepsilon_{2}$ simultaneously. Subsequently, the optimal adaptive learning rate $\eta_{i}(t)$, $t \!\in\! \{0, 1, \ldots, T\}$ can be derived based on the fixed point of $\phi_{1}(t)$ and $\phi_{2}(t)$, $t \!\in\! \{0, 1, \ldots, T\}$, as summarized in Fig. \ref{framework}.

\item \textbf{Step 3 (Data Collection and Local Training):} Each participating client gathers local data from their geographically discrete mobile devices, retrieves the updated global model from the central server, and subsequently conducts local training on their private database $\mathcal{D}_{i}$ with respective optimal adaptive learning rate $\eta_{i}(t)$.

\item \textbf{Step 4 (Parameter Upload):} Upon accomplishing the local training process, each client uploads the local model parameter to the central server for model aggregation in conformity with the predefined transmitting rules.

\item \textbf{Step 5 (Model Aggregation):} Once receiving the local parameters from the participating clients, the central server derives a global model matrix for local model aggregation and updates the global model. Furthermore, the central server assesses the quality of the current global model via a pre-partitioned testing dataset. If the test accuracy exceeds the required threshold or the number of global iterations reaches the predetermined limitation, the federated training process concludes. Otherwise, the central server disseminates the latest global model parameters to each client for the subsequent iteration of global model training by repeating the FL model training Step 3-5.
\end{itemize}

\subsection{Adaptive Learning Rate for Each Client}\label{subsec_adp_lr_client}

Note that global model $\boldsymbol{w}(t) \!=\! \sum_{i=1}^{N} \theta_i \boldsymbol{w_i}(t)$ and the component $\sum_{i=1}^{N}  \theta_{i} \boldsymbol{w_{i}}(t)^{T}\boldsymbol{w_{i}}(t)$ in Eq. (\ref{p_init}) are both the integration of local updated parameters obtained by participating clients after performing local training. The learning rate will affect the local updated parameter of the individual client, as well as the value of $\sum_{i=1}^{N} \theta_i \boldsymbol{w_i}(t)$ and $\sum_{i=1}^{N}  \theta_{i} \boldsymbol{w_{i}}(t)^{T}\boldsymbol{w_{i}}(t)$, which in turn affects the adaptive learning rate design. However, the local information of other clients is unknown due to the communication restrictions. Thus, we introduce the following mean-field estimators to evaluate other clients' local information and enable decentralized learning rate design.
\begin{definition}\label{def2}
To design the decentralized adaptive learning rate for each participating client, two mean-field terms are introduced to estimate the global parameter $\sum_{i=1}^{N} \theta_i \boldsymbol{w_i}(t)$ and the component $\sum_{i=1}^{N}  \theta_{i} \boldsymbol{w_{i}}(t)^{T}\boldsymbol{w_{i}}(t)$ respectively, i.e.,
\begin{align}
\boldsymbol{\phi_1}(t)  &= \textstyle \sum_{i=1}^{N} \theta_i \boldsymbol{w_i}(t), \label{phi_1} \\
\phi_2(t) &= \textstyle \sum_{i=1}^{N} \theta_{i} \boldsymbol{w_{i}}(t)^{T}\boldsymbol{w_{i}}(t). \label{phi_calc}
\end{align}
\end{definition}

Thus, on the basis of Definition \ref{def2}, the entropy term defined in Eq. (\ref{p_init}) can be rewritten as follows:
\begin{align}\label{p_phi}
\bar{p}_{i}(t) = \frac{\theta_{i} \boldsymbol{w_{i}}(t)^{T}\boldsymbol{w_{i}}(t)}{\phi_{2}(t)} \in [0, 1].
\end{align}

Accordingly, the objective function of participating client $i \! \in \! \{1,2, \ldots,N\}$ can be reconstructed as follows:
\begin{align}
J_{i}(T) \! =& \min_{\substack{\eta_{i}(t)}}\! \sum_{t=0}^{T}\! \left(\!\beta \! \sum_{j=1}^{N} \bar{p}_{j}(t) \log (\bar{p}_{j}(t)) \!+\! (1 \!-\! \beta)\eta^{2}_{i}(t)\!\!\right), \label{new_objective} \\ 
&\text{s.t.}~~~ \boldsymbol{w_i}(t+1) = \boldsymbol{\phi_1}(t)-\eta_i(t)\nabla F_i(\boldsymbol{\phi_1}(t)). \label{equ_dynamicsMF} 
\end{align}

It is obvious that in the above Eqs. (\ref{new_objective})-(\ref{equ_dynamicsMF}), the adaptive learning rate of each client $i$ is only affected by its local model information and the mean-field estimators $\boldsymbol{\phi_1}(t), \phi_2(t)$, which can be obtained as shown in Section \ref{sec_4.3}.

The problem in Eqs. (\ref{new_objective})-(\ref{equ_dynamicsMF}) is a discrete-time nonlinear system optimization problem with dynamic constraints, which is challenging to solve analytically due to the curse of dimensionality. Imaging that a significant quantity of adaptive learning rate $\eta_{i}(t) \! \in \! [0,1]$, $t \! \in \! \{0,1,\ldots, T\}$ must be jointly designed with local model information, which will in turn affect the process of local parameter updating and adaptive learning rate design. The computation complexity $O((N \! + \! M)^{T})$ is formidably high and increases exponentially in $T$, where the computation complexity of the gradient of nonlinear loss function $\nabla F_i(*)$ is defined by $O(M)$. By constructing the Hamiltonian equation, we derive the closed-form of adaptive learning rate $\eta_i(t)$ as
summarized in the following proposition.


\begin{pro}\label{pro_eta_t_0}
Given the mean-field terms $\boldsymbol{\phi_{1}}(t)$ and $\phi_{2}(t)$ in Definition \ref{def2} with $t \! \in \! \{0,1,\ldots,T\}$, the optimal adaptive learning rate $\eta_{i}(t)$ of client $i \! \in \! \{1,2,\ldots,N\}$ at iteration $t$ satisfies:
\begin{align}\label{eta_solve_t_1}
\eta_{i}(t) &= \left(\frac{\theta_{i}(\boldsymbol{\phi_{1}}(t) - \eta_{i}(t)\nabla F_{i}(\boldsymbol{\phi_{1}}(t)))^{T}}{\phi_{2}(t + 1)} \times \frac{\beta \nabla F_{i}(\boldsymbol{\phi_{1}}(t))}{(1 - \beta)} \right) \nonumber \\ 
& \times \left[1 + \log \left(  \theta_{i} (\boldsymbol{\phi_{1}}(t) - \eta_{i}(t)\nabla F_{i}(\boldsymbol{\phi_{1}}(t)))^{T} \right. \right. \nonumber \\
& \left. \left. \times \frac{(\boldsymbol{\phi_{1}}(t) - \eta_{i}(t)\nabla F_{i}(\boldsymbol{\phi_{1}}(t)))}{\phi_{2}(t + 1)} \right) \right],
\end{align}
which is upper bounded by:
\begin{align} \label{eta_upper_bound}
\eta_{i}(t)  \leq \frac{ 2\beta \theta_{i} D \| \boldsymbol{\phi_{1}}(t) \| }{(1 - \beta) \phi_{2}(t + 1)},
\end{align}
with $\eta_i(T)=0$.
\end{pro}
\begin{proof} Based on the objective function as presented in Eq. (\ref{new_objective}) with mean-field terms $\boldsymbol{\phi_{1}}(t)$ and $\phi_2(t)$, we can construct the following discrete-time Hamilton equation:
\begin{align}\label{Hamilton_eq}
 H(t) &= \beta \sum_{j=1}^{N}\bar{p}_{j}(t) \log \bar{p}_{j}(t) + (1 \!-\! \beta)\eta^{2}_{i}(t)  \nonumber \\
& +\! \boldsymbol{\lambda_{i}}(t \!+\! 1)(\boldsymbol {\phi_{1}(t)} \!-\! \eta_{i}(t)\nabla F_{i}(\boldsymbol{\phi_{1}}(t)) \!-\! \boldsymbol{w_{i}}(t)), 
\end{align}
where $\boldsymbol{\lambda_{i}}(t + 1) \! \in \! \mathbb{R}^{d}$ is a vector whose size is $1 \! \times \! d$. Since $\frac{\partial^2 H(t)}{\partial \eta_{i}^2(t)} \! = \! 2(1 \! - \! \beta) \! > \! 0$, the Hamiltonian function $H(t)$ is convex in $\eta_{i}(t)$. Therefore, to obtain the optimal adaptive learning rate to minimize the objective function $J_{i}(T)$, it is necessary to satisfy the following properties \cite{di2012discrete}:
\begin{align}\label{eta_first_derivative}
\frac{\partial H(t)}{\partial \eta_{i}(t)} = 0,
\boldsymbol{\lambda_{i}}(t + 1) - \boldsymbol{\lambda_{i}}(t) = -\frac{\partial H(t)}{\partial \boldsymbol{w_{i}}(t)}.
\end{align}

Based on Eq. (\ref{eta_first_derivative}), we can first acquire an expression of adaptive learning rate $\eta_{i}(t)$:
\begin{align} \label{eta_solve_1}
2(1 - \beta)\eta_{i}(t) &+ \boldsymbol{\lambda_{i}}(t + 1)(-\nabla F_{i}(\boldsymbol{\phi_{1}}(t))) = 0,  \nonumber  \\
\Rightarrow \  \eta_{i}(t) &= \frac{\boldsymbol{\lambda_{i}}(t + 1)\nabla F_{i}(\boldsymbol{\phi_{1}}(t))}{2(1 - \beta)}. 
\end{align}

Then, we derive the expression of impulse vector $\boldsymbol{\lambda_{i}}(t)$:
\begin{small}
\begin{align} \label{lambda_solve}
&\boldsymbol{\lambda_{i}}(t + 1) - \boldsymbol{\lambda_{i}}(t) \nonumber \\
&= \boldsymbol{\lambda_{i}}(t + 1) \!-\! \beta\frac{2\theta_{i} \boldsymbol{w_{i}}(t)^{T}}{\phi_{2}(t)} \left(\log \left(\frac{\theta_{i} \boldsymbol{w_{i}}(t)^{T}\boldsymbol{w_{i}}(t)}{\phi_{2} (t)}\right) + 1  \right) , \nonumber \\ 
&\Rightarrow \ \boldsymbol{\lambda_{i}}(t) \!=\! \beta\frac{2\theta_{i} \boldsymbol{w_{i}}(t)^{T}}{\phi_{2}(t)} \left(\log \left(\frac{\theta_{i} \boldsymbol{w_{i}}(t)^{T}\boldsymbol{w_{i}}(t)}{\phi_{2} (t)}\right) \!+\! 1  \right) .
\end{align}
\end{small}

By inserting $\boldsymbol{\lambda_{i}}(t)$ into the expression of $\eta_{i}(t)$ in Eq. (\ref{eta_solve_1}), 
the optimal adaptive learning rate of client $i$  can be derived as Eq. (\ref{eta_solve_t_1}).
Based on Eq. (\ref{eta_solve_t_1}), we can obtain:
\begin{align} \label{eta_and_nabla_1}
& \eta_{i}(t) = \left( \frac{\beta \theta_{i}}{1- \beta} \cdot \frac{ (\boldsymbol{\phi_{1}}(t) - \eta_{i}(t) \nabla F_{i}(\boldsymbol{\phi_{1}}(t)) )^{T} \nabla F_{i}(\boldsymbol{\phi_{1}}(t))}{\phi_{2}(t + 1)} \right)  \nonumber \\
& \qquad\qquad \times \left(\log ( p_{i}(t + 1)) + 1  \right)  \\
& \Rightarrow \ \eta_{i}(t) \!=\! \max \left\{0,  \frac{\boldsymbol{\phi_{1}}(t)^{T} \nabla F_{i}(\boldsymbol{\phi_{1}}(t))}{\frac{(1 - \beta)\phi_{2}(t + 1)}{\beta \theta_{i} (1 + \log p_{i}(t + 1))}  \!+\! \| \nabla F_{i}(\boldsymbol{\phi_{1}}(t)) \|} \right\}.
\end{align}

Note that $\eta_{i}(t) \geq 0$, $0 < \theta \leq 1$, $0 <\beta \leq 1$,  $1 + \log(p_{i}) \leq 2, \phi_{2}(t) \geq 0$, $\| \nabla F_{i}(\boldsymbol{\phi_{1}}(t)) \| \geq 0$, thus we can get:
\begin{align}
\eta_{i}(t) & \leq 
\frac{ 2\beta \theta_{i}  | \boldsymbol{\phi_{1}}(t)^{T} \nabla F_{i}(\boldsymbol{\phi_{1}}(t)) |}{(1 - \beta)\phi_{2}(t + 1) } \overset{(a)}{\leq} \frac{ 2\beta \theta_{i} \| \boldsymbol{\phi_{1}}(t) \| \| \nabla F_{i}(\boldsymbol{\phi_{1}}(t)) \| }{(1 - \beta)\phi_{2}(t + 1) } \nonumber \\ 
& \overset{(b)}{\leq} \frac{ 2\beta \theta_{i} D \| \boldsymbol{\phi_{1}}(t) \| }{(1 - \beta)\phi_{2}(t + 1)},
\end{align}
where $(a)$ follows inequality $(\boldsymbol{a}, \boldsymbol{b}) \leq \| \boldsymbol{a} \| \cdot \| \boldsymbol{b} \|$, $(b)$ follows Assumption \ref{Bounded_Gradients}.  \end{proof}

Proposition \ref{pro_eta_t_0} shows that the adaptive learning rate $\eta_{i}(t)$ is barely associated with the linear combination of the predefined hyperparameters and mean-field terms $\boldsymbol{\phi_1}(t)$ and $\phi_2(t)$, which can be viewed as a given function. Further, the upper bound of the adaptive learning rate guarantees stability and avoids wide fluctuation during the process of model training.

\begin{algorithm}[t] 
\caption{Iterative Computation of Fixed Point} 
\begin{algorithmic}[1] 
\REQUIRE ~~\\ 
Initiate the fixed point threshold $\varepsilon_{1}=\varepsilon_{2}= 0.001$, global model parameters $\boldsymbol{w}(0)$, mean-field estimators $\boldsymbol{\phi_{1}^{0}}(t)$, $\phi_{2}^{0}(t)$, $t\in\{1,2,\ldots,T\}$, number of iterations $k=0$. 
\ENSURE ~~\\ 
The fixed point of the mean-field items $\boldsymbol{\phi_{1}}(t)$ and $\phi_{2}(t)$, $t \in \{1, 2, \ldots, T\}$.

\STATE \textbf{do}
    \FOR{$t=0$ to $T-1$} 
        \FOR{$i=1$ to $N$}
            \STATE compute $\eta_i(t)$ based on $\boldsymbol{\phi_{1}^k}(t), \phi_{2}^{k}(t)$, $t\!\in\!\{1,2,\ldots,$
            \STATE $T\}$ according to Eq. (\ref{eta_solve_t_1})
            \STATE compute $\boldsymbol{w_i}(t+1)$ according to Eq. (\ref{equ_dynamicsMF})
        \ENDFOR
        \STATE compute $\boldsymbol{\phi_{1}^{k+1}}(t+1)$ according to Eq. (\ref{phi_1})
        \STATE compute $\phi_{2}^{k+1}(t+1)$ according to Eq. (\ref{phi_calc})
    \ENDFOR 
\STATE $k=k+1$
\STATE \textbf{while} $\boldsymbol{\phi_{1}^{k}}(t)$-$\boldsymbol{\phi_{1}^{k-1}}(t)$ $\geq$ $\varepsilon_{1}$ and $\phi_{2}^{k}(t) - \phi_{2}^{k-1}(t) \geq  \varepsilon_{2}$

\end{algorithmic}
\label{alg of the fixed point} 
\end{algorithm}

\subsection{Update of Mean-Field Estimators for Finalizing the Adaptive Learning Rate}\label{sec_4.3}
In this section, we determine the mean-field terms $\boldsymbol{\phi_{1}}(t)$ and $\phi_{2}(t)$ according to the given optimal adaptive learning rate $\eta_{i}(t)$ in Eq. (\ref{eta_solve_t_1}). Note that the mean-field terms $\boldsymbol{\phi_{1}}(t)$ and $\phi_{2}(t)$ constructed to estimate $\boldsymbol{w}(t)$ and $\sum_{i=1}^{N} \theta_{i}\boldsymbol{w_{i}}(t)^{T}\boldsymbol{w_{i}}(t)$ are affected by the local updated parameters $\boldsymbol{w_{i}}(t)$ of all clients, which will in turn affect the local parameters. In the following proposition, we will determine the appropriate mean-field terms $\boldsymbol{\phi_{1}}(t)$ and $\phi_{2}(t)$ by fixed point theory.
\begin{pro}\label{pro_fix_point}
There exists fixed points for the mean-field estimators $\boldsymbol{\phi_{1}}(t)$ and $\phi_{2}(t)$ respectively, where $t \! \in \! \{0, 1, \\ \ldots, T - 1\} $.
\end{pro}
\begin{proof}
For any client $i \! \in \! \{1,2, \ldots, N\}$, substitute $\boldsymbol{\phi_{1}}(t) \! = \! \sum_{i=1}^{N} \theta_{i} \boldsymbol{w_{i}}(t)$, $\phi_{2}(t) \! = \! \sum_{i=1}^{N} \theta_i \boldsymbol{w_i}(t)^T\boldsymbol{w_i}(t)$ and $\eta_{i}(t)$ into Eq. (\ref{equ_dynamicsMF}),  we can observe that the local parameter $\boldsymbol{w_{i}}(t)$ of client $i$ at global training iteration $t$ is a function of the model parameters $\{\boldsymbol{w_{i}}(t)|t \! \in \! \{0, 1, \ldots, T\}, i \in \{1,2, \ldots, N\}  \}$ of all zones over time. Define the following function by a mapping from $\{\boldsymbol{w_{i}}(t)|t \! \in \! \{0, 1, \ldots, T\}, i \in \{1,2, \ldots, N\} \}$ to the model parameter $\boldsymbol{w_{i}}(t)$ of client $i$ in Eq. (\ref{eta_solve_t_1}) at global iteration $t$, i.e.,
\begin{align}\label{gamma_w_i_t}
\small \Gamma_{i, t}(\{ \boldsymbol{w_{i}}(t)|t  \!\in\!   \{0, 1, \ldots, T\}, i  \!\in\!  \{1,2, \ldots, N\} \}) \!=\! \boldsymbol{w_{i}}(t).
\end{align}

To summarize any possible mapping $\Gamma_{i, t}(*)$ in Eq. (\ref{gamma_w_i_t}), we define the following vector function as a mapping from $\{\boldsymbol{w_{i}}(t)|t \! \in \! \{0, 1, \ldots, T\}, i \! \in \! \{1,2, \ldots, N\}  \}$ to the set of all the clients' model parameters over time, i.e.,
\begin{align}\label{gamma_w_i_t_mapping}
& \Gamma(\{\boldsymbol{w_{i}}(t)|t \! \in \! \{0, 1, \ldots, T\}, i \! \in \! \{1,2, \ldots, N\} \})   \nonumber \\
& = (\Gamma_{1, 0}(\{\boldsymbol{w_{i}}(t)|t \! \in \! \{0, 1, \ldots, T\}, i \! \in \! \{1,2, \ldots, N\} \}), \ldots,  \nonumber \\
& \quad  \Gamma_{1, T}(\{\boldsymbol{w_{i}}(t)|t \! \in \! \{0, 1, \ldots, T\}, i \! \in \! \{1,2, \ldots, N\} \}), \ldots, \nonumber \\
& \quad  \Gamma_{N, 0}(\{\boldsymbol{w_{i}}(t)|t \! \in \! \{0, 1, \ldots, T\}, i \! \in \! \{1,2, \ldots, N\} \}), \ldots,  \nonumber \\
& \quad  \Gamma_{N, T}(\{\boldsymbol{w_{i}}(t)|t \! \in \! \{0, 1, \ldots, T\}, i \! \in \! \{1,2, \ldots, N\} \})).
\end{align}

Thus, the fixed point of mapping $\Gamma(\{\boldsymbol{w_{i}}(t)|t \!\in\! \{0, 1, \ldots, \\ T\}, i \!\in\! \{1,2, \ldots, N\} \})$ should be reached to make $\boldsymbol{\phi_{1}}(t)$ and $\phi_{2}(t)$ replicate $\sum_{i=1}^{N} \theta_{i} \boldsymbol{w_{i}}(t)$ and $\sum_{i=1}^{N} \theta_i \boldsymbol{w_i}(t)^T\boldsymbol{w_i}(t)$, respectively.

Note that $\boldsymbol{w_{i}}(0)$ worked as the initial value of the local model parameters should be bounded for any client $i$. Thus, $\boldsymbol{\phi_{1}}(0) \! = \! \sum_{i=1}^{N} \theta_{i} \boldsymbol{w_{i}}(0)$ is bounded as well. According to Eq. (\ref{eta_solve_t_1}), we can derive: $\eta_{i} (t) \!=\! A_{i}(\boldsymbol{\phi_{1}}(t), \phi_2(t + 1))$, where $A_{i}(*)$ is the analytic function of $\eta_{i}(t)$. Then, we can obtain:
\begin{align} \label{phi_next_t_2}
\phi_{2}(t \!+\! 1)  \!&=\! \sum_{i=1}^{N}[\theta_{i} (\boldsymbol{\phi_{1}}(t) \!-\! A_{i}(\boldsymbol{\phi_{1}}(t), \phi_2(t \!+\! 1))\nabla F_{i} (\boldsymbol{\phi_{1}}(t)))^{T} \nonumber \\
&  \times  (\boldsymbol{\phi_{1}}(t) \!-\! A_{i}(\boldsymbol{\phi_{1}}(t), \phi_2(t \!+\! 1))\nabla \! F_{i} (\boldsymbol{\phi_{1}}(t)))]. \!
\end{align}

Based on Eq. (\ref{phi_next_t_2}), we define $B(*)$ as the analytic function of $\phi_{2}(t)$, i.e., $\phi_{2} (t + 1) = B(\boldsymbol{\phi_{1}}(t))$, then we can derive the expression of learning rate $\eta_{i} (t)$ as follows:
\begin{align}\label{eta_solove_based_on_phi_1}
    \eta_{i} (t) = A_{i}(\boldsymbol{\phi_{1}}(t), B(\boldsymbol{\phi_{1}}(t))).
\end{align}

Since $A_{i}(*)$ and $B(*)$ are the analytic function, and based on Eq. (\ref{eta_solove_based_on_phi_1}), we assume that the maximum upper bound of adaptive learning rate $\eta_{i}(t)$ is $C$, i.e., $\eta_{i}(t) \leq  C$. Then, according to Assumption \ref{Bounded_Gradients}, we have:
\begin{align}\label{local_grad_bound}
    \| \nabla F_{i}(\boldsymbol{w_{i}}(t)) \| \leq D.
\end{align}

By inserting Eq. (\ref{local_grad_bound}) into Eq. (\ref{equ_dynamicsMF}), we can get:
\begin{align}
& \quad\ \ \boldsymbol{w_{i}}(1) = \boldsymbol{\phi_{1}}(0)  - \eta_{i}(0)\nabla F_{i}(\boldsymbol{\phi_{1}}(0)), \\
\Rightarrow \ & \left\{\begin{array}{l}
\boldsymbol{\phi_{1}}(0) - CD \leq  \boldsymbol{w_{i}}(1) \leq \boldsymbol{\phi_{1}}(0) + CD, \\
\boldsymbol{\phi_{1}}(0) - CD \leq  \boldsymbol{\phi_{1}}(1) \leq \boldsymbol{\phi_{1}}(0) + CD, \\
\quad\quad\quad\quad\quad\quad\quad\quad \vdots \\
\boldsymbol{\phi_{1}}(0) - TCD \leq  \boldsymbol{\phi_{1}}(t) \leq \boldsymbol{\phi_{1}}(0) + TCD. 
\end{array}\right. 
\end{align}

Define set $\Omega = [\boldsymbol{\phi_{1}}(0) - CD, \boldsymbol{\phi_{1}}(0) + CD] \times \ldots \times  [\boldsymbol{\phi_{1}}(0) - TCD, \boldsymbol{\phi_{1}}(0) + TCD]$. Since $\Gamma_{i, t}$ is continuous, $\Gamma$ is also a continuous mapping from $\Omega$ to $\Omega$. According to Brouwer’s fixed point theorem, $\Gamma$ has a fixed point in $\Omega$. 
\end{proof}

\begin{algorithm}[t]
\caption{Federated Learning with Adaptive Learning Rate via Entropy (FedEnt)}
\begin{algorithmic}[1]
\REQUIRE ~~\\ 
    Initiate the global model parameter $\boldsymbol{w}(0)$, mean-filed terms $\boldsymbol{\phi_1}(t), \phi_{2}(t), t \in \{0, 1, .., T\}$ as returned by Algorithm \ref{alg of the fixed point}, learning rate $\eta_{i}(-1)$, control parameter $\beta$, decay parameter $\gamma$ and the local training epoch of the client $E$.
\ENSURE ~~\\ 
    The optimal global model parameter $\boldsymbol{w}(t)$.
\FOR {$t=0,\ldots,T-1$}
    \STATE  $\boldsymbol{w_{i, 0}}(t) = \boldsymbol{w}(t)$
    \FOR {each client $i \in \{1, 2, \ldots, N\}$ \textbf{in parallel}}
        \STATE Sample a random batch of data from $\mathcal{D}_i$ to compute the local gradient $\nabla F_{i}(\boldsymbol{w}(t))$
        \STATE Solve $\eta_{i}(t)$ according to Eq. (\ref{eta_solve_t_1}), (\ref{eta_deacy})
        \FOR {$k=0, \ldots, E - 1$}
            \STATE Compute the local gradient $\nabla F_{i}(\boldsymbol{w_{i, k}}(t))$
            \STATE  $\boldsymbol{w_{i, k + 1}}(t) = \boldsymbol{w_{i, k}}(t) - \eta_{i}(t) \nabla F_{i}(\boldsymbol{w_{i, k}}(t))$
        \ENDFOR
    \ENDFOR 
    \STATE $\boldsymbol{w}(t + 1) = \sum_{i=1}^{N} \theta_{i}  \boldsymbol{w_{i, E}}(t + 1)$
\ENDFOR
\end{algorithmic}
\label{alg_adlr}
\end{algorithm}

The pretraining process of solving the fixed points in Algorithm \ref{alg of the fixed point} is with the linear computation complexity $O(TNK)$, where $K$ represents the number of iterations acquired for calculating the fixed point of Eqs. (\ref{phi_1})-(\ref{phi_calc}), thus avoiding introducing much additional time costs during the formal FL model training. Given arbitrary initial global model parameters $\boldsymbol{w}(0)$ and  $k$-th mean-field estimators $\boldsymbol{\phi_{1}^{k}}(0)$ and $\phi_{2}^{k}(0)$, we can easily solve $\eta_{i}(0)$ and update the local model parameters $\boldsymbol{w_i(0)}, i\in\{1,2, \ldots, N\}$. Then we can get $\boldsymbol{\phi_{1}^{k + 1}}(1)$, $\phi_{2}^{k + 1}(1)$ based on Eqs. (\ref{phi_1})-(\ref{phi_calc}). Similar to the prior researches on multi-objective optimization problem \cite{grandoni2014new, grandoni2009iterative, yang2021adaptive}, repeating the process above until $\boldsymbol{\phi_{1}^{k}}(t) \!-\! \boldsymbol{\phi_{1}^{k-1}}(t)$ and $\phi_{2}^{k}(t) \!-\! \phi_{2}^{k-1}(t)$ simultaneously reach the arbitrarily small error $\varepsilon_{1}$ and $\varepsilon_{2}$, which is served as a clipping threshold that is infinitely close to zero (usually defaulted by $10^{-3}$\cite{tu2022adaptive}), thus the fixed point is found. Further, according to the proof of Proposition \ref{pro_fix_point}, the fixed points of the mean-field terms $\boldsymbol{\phi_{1}}(t)$ and $\phi_{2}(t)$, $t \!\in\! \{0, 1, \ldots, T\}$ replicate $\sum_{i=1}^{N} \theta_{i} \boldsymbol{w_{i}}(t)$ and $\sum_{i=1}^{N} \theta_i \boldsymbol{w_i}(t)^T\boldsymbol{w_i}(t)$ respectively, which finalizes the derivation of the adaptive learning rate according to Proposition~\ref{pro_eta_t_0}.

Based on the Algorithm \ref{alg of the fixed point}, we obtain the mean-filed terms $\boldsymbol{\phi_1}(t)$ and $\phi_{2}(t)$, $t \! \in \! \{0, 1, \ldots, T\}$ for calculating the optimal learning rate $\eta_{i}(t)$. In addition, we introduce the following weight decay method:
\begin{align} \label{eta_deacy}
\eta_{i}(t) = \gamma * \eta_{i}(t - 1) + (1 - \gamma) * \eta_{i}(t),
\end{align}
where $\gamma \in (0, 1)$ represents the decay hyperparameter that controls the aggregation weight and avoids the excessive fluctuations on the learning rate $\eta_{i}(t)$. Referring to FedAvg, the process of our proposed adaptive federated learning algorithm FedEnt is summarized in Algorithm \ref{alg_adlr} by considering $E$ local training epochs conducted by each participating client in parallel at each global iteration.

\section{Convergence Analysis for FedEnt}\label{sec_conv_an}
In this section, we demonstrate the theoretical analysis for our proposed adaptive FL algorithm FedEnt and analyze the convergence property.

\subsection{Bounding Client Drifting} \label{subsec_bound_w}
Previous research indicates that FedAvg suffers from client drifting with heterogeneous data (i.e., Non-IID data), resulting in model instability and slow convergence rate \cite{karimireddy2020scaffold}. To verify the effectiveness of our proposed adaptive federated learning algorithm FedEnt in alleviating client drifting, we first derive the upper bound of global parameter deviation as summarized in the following proposition.
\begin{pro}\label{thm_w_gap_bound}
The difference between the global parameter $\boldsymbol{w}(t)$ and $\boldsymbol{w}(t + 1)$ in any two consecutive global iterations is bounded, i.e., 
\begin{align}\label{w_gap_bound}
\| \boldsymbol{w}(t \!+\! 1) \!-\! \boldsymbol{w}(t) \| &\!\leq\! \sum_{i=1}^{N} \frac{2 \beta \theta_{i}^{2}  D \| \boldsymbol{\phi_{1}}(t) \| \| \nabla F_{i } (\boldsymbol{w}(t)) \|}{(1 - \beta)\phi_{2}(t + 1)} \nonumber \\
&\!\leq\! \sum_{i=1}^{N} \frac{2 \beta \theta_{i}^{2}  D^{2} \| \boldsymbol{\phi_{1}}(t) \|}{(1 - \beta)\phi_{2}(t + 1)}.
\end{align}
\end{pro}
\begin{proof}
According to the global model aggregation rule in Eq. (\ref{w_global_static_update}), local model update rule in Eq. (\ref{dynamic_w_update}) and the upper bound of the adaptive learning rate $\eta_{i}(t)$ in Proposition \ref{pro_eta_t_0}, we can easily derive as follows:
\begin{align}
\| \boldsymbol{w}(t + 1) - \boldsymbol{w}(t) \| &=  \textstyle \| \sum_{i=1}^{N} \theta_{i} \eta_{i}(t) \nabla F_{i } (\boldsymbol{w}(t)) \|  \nonumber \\
& \overset{(a)}{\leq} \textstyle \sum_{i=1}^{N} \theta_{i}  |\eta_{i}(t)| \| \nabla F_{i } (\boldsymbol{w}(t)) \| \nonumber\\
& \overset{(b)}{\leq} \sum_{i=1}^{N} \frac{2 \beta \theta_{i}^{2}  D \| \boldsymbol{\phi_{1}}(t) \| \| \nabla F_{i } (\boldsymbol{w}(t)) \|}{(1 - \beta)\phi_{2}(t + 1)} \nonumber \\ 
& \leq \sum_{i=1}^{N} \frac{2 \beta \theta_{i}^{2}  D^{2} \| \boldsymbol{\phi_{1}}(t) \|}{(1 - \beta)\phi_{2}(t + 1)},
\end{align}
where $(a)$ follows triangle inequality $\|\boldsymbol{a} \!+\! \boldsymbol{b}\| \!\leq\! \|\boldsymbol{a}\| \!+\! \|\boldsymbol{b}\|$, $(b)$ follows the upper bound of the learning rate in Eq. (\ref{eta_upper_bound}). \end{proof}

Proposition \ref{thm_w_gap_bound} demonstrates that the upper bound of the difference between the global model parameter in two successive global iterations exists and is affected by the known mean-field terms $\boldsymbol{\phi_{1}}(t)$ and $\phi_{2}(t)$. Moreover, we define the upper bound of the divergence between the global parameter $\boldsymbol{w}(t)$ and $\boldsymbol{w}(t + 1)$ as follows:
\begin{align} \label{G}
G(\| \boldsymbol{w}(t \!+\! 1) \!-\! \boldsymbol{w}(t) \|) \!=\! \sum_{i=1}^{N} \frac{2 \beta \theta_{i}^{2}  D^{2} \| \boldsymbol{\phi_{1}}(t) \|}{(1 - \beta)\phi_{2}(t + 1)}. 
\end{align}
 
Then, based on the upper bound of the divergence between the global parameter $\boldsymbol{w}(t)$ and $\boldsymbol{w}(t + 1)$ in Proposition \ref{thm_w_gap_bound} and the proposed function $G$ in Eq. (\ref{G}), we can derive the client drifting as following.
\begin{pro} \label{thm_w_global_local_bound}
The difference between the global parameter $\boldsymbol{w}(t)$ and local parameter $\boldsymbol{w_{i}}(t), i \in \{1,2, \ldots, N\}$ in $t$-th global iteration is
\begin{align}\label{w_local_global}
\| \boldsymbol{w}(t) \!-\! \boldsymbol{w_{i}}(t) \| & \!<\!  \sum_{j=1}^{N} \! \frac{2 \beta \theta_{j}^{2} D^{2} \| \boldsymbol{\phi_{1}}(t \!-\! 1) \|}{(1 \!-\! \beta)\phi_{2}(t)} \nonumber \\
& \!=\! G(\| \boldsymbol{w}(t) - \boldsymbol{w}(t - 1) \|).
\end{align}
\end{pro}
\begin{proof}
Based on Eq. (\ref{w_global_static_update}) and Eq. (\ref{equ_dynamicsMF}), we have:
\begin{align}
& \| \boldsymbol{w}(t) \!-\! \boldsymbol{w_{i}}(t) \| \!=\! \| \boldsymbol{w}(t) \!-\! \boldsymbol{\phi_{1}}(t \!-\! 1) \!+\! \theta_{i} \eta_{i}(t \!-\! 1) \nabla \! F_{i}(\boldsymbol{\phi_{1}}(t \!-\! 1)) \| \nonumber \\
&= \textstyle \| \theta_{i} \eta_{i}(t \!-\! 1) \nabla F_{i}(\boldsymbol{\phi_{1}}(t \!-\! 1)) \!-\! \sum_{j=1}^{N} \theta_{j}\eta_{j}(t \!-\! 1) \nabla \! F_{j} (\boldsymbol{\phi_{1}}(t \!-\! 1)) \| \nonumber \\
& = \textstyle \| \!-\! \sum_{j=1, j \neq i}^{N} \theta_{j}\eta_{j}(t \!-\! 1) \nabla F_{j} (\boldsymbol{\phi_{1}}(t \!-\! 1)) \| \nonumber \\ 
& <  \sum_{j=1}^{N} \! \frac{2 \beta \theta_{j}^{2} D^{2} \| \boldsymbol{\phi_{1}}(t \!-\! 1) \|}{(1 \!-\! \beta)\phi_{2}(t)} \!=\! G(\| \boldsymbol{w}(t) \!-\! \boldsymbol{w}(t \!-\! 1) \|). \nonumber  
\end{align}

Hence, the proposition holds. \end{proof}

According to Proposition \ref{thm_w_global_local_bound}, it is evident that the upper bound of the difference between the global model parameter $\boldsymbol{w}(t)$ and the local model parameter $\boldsymbol{w_{i}}(t)$ at $t$-th global iteration converges to the linear combination of data weight $\theta_{i}$, mean-field terms $\boldsymbol{\phi_{1}}(t), \phi_{2}(t)$ and the predefined penalty hyperparameter $\beta$. Further, the difference between the local and global model parameters is smaller than that in the previous iteration, which verifies the effectiveness of FedEnt in mitigating the negative influence of client drifting over~time.

\subsection{Bounding Global Loss Function}\label{subsec_bound_F}
According to the upper bound of the client drifting in Proposition \ref{thm_w_global_local_bound}, we derive the following theorem regarding the upper bound of the global loss function.
\begin{thm}\label{thm_F_global_bound}
For any two consecutive global iterations $t$ and $t + 1$, the convergence upper bound of FedEnt 
satisfies,
\begin{align} \label{lobal_gradient_bound_final}
\footnotesize F(\boldsymbol{w}(t \!+\! 1)) \!\!-\!\! F(\boldsymbol{w}(t)) & \footnotesize \!\leq\! \frac{L}{2}\!\!\left[\sum_{i=1}^{N} \!\! \frac{2 \beta \theta_{i}^{2} \! D \|\boldsymbol{\phi_{1}}\!(t)\| \| \nabla \!F_{i} (\boldsymbol{w}(t)) \|}{(1 - \beta)\phi_{2}(t + 1) }\!\right]^{2} \!\!\!\!-\!\! D \nonumber \\
& \footnotesize \!\leq\! \frac{L}{2}\!\!\left[\sum_{i=1}^{N} \!\! \frac{2 \beta \theta_{i}^{2} \! D^{2} \|\boldsymbol{\phi_{1}}\!(t)\|}{(1 - \beta)\phi_{2}(t + 1) }\!\right]^{2} \!\!\!\!-\!\! D,\end{align}
which is smaller than that of the FedAvg. 
\end{thm}
\begin{proof}
For the central server, the global model update rule at $t$-th training iteration can be rewritten as follows:
\begin{align}\label{global_w_update}
\boldsymbol{w}(t + 1) = \boldsymbol{w}(t) - \eta(t) \nabla F(\boldsymbol{w}(t)), 
\end{align}
where $\eta(t) \! = \! \sum_{i=1}^{N} \theta_{i} \eta_{i}(t)$ represents the learning rate of the central server at $t$-th global training iteration.

Based on Eq. (\ref{global_w_update}), we can easily acquire the expression of the gradient of the global loss function:
\begin{align}\label{eq_label}
\nabla F(\boldsymbol{w}(t)) &= \frac{\boldsymbol{w}(t) - \boldsymbol{w}(t + 1)}{\eta(t) }.
\end{align}

From the $L$-Lipschitz smoothness of $F(\boldsymbol{w}(t))$ in Assumption \ref{Lipschitz_Smoothness} and Taylor expansion, we have:
\begin{align}\label{global_gradient_bound}
& F(\boldsymbol{w}(t + 1)) - F(\boldsymbol{w}(t))  \nonumber \\
&\leq \left< \nabla F(\boldsymbol{w}(t)), \boldsymbol{w}(t + 1)\! - \! \boldsymbol{w}(t) \right> \!  + \! \frac{L}{2} {\| \boldsymbol{w}(t + 1) \! - \!\boldsymbol{w}(t) \|}^2 \nonumber \\
& \leq - \frac{\| \boldsymbol{w}(t + 1) - \boldsymbol{w}(t) \| }{\eta (t)} +  \frac{L}{2} {\| \boldsymbol{w}(t + 1) - \boldsymbol{w}(t) \|}^2 . 
\end{align}

Based on Eq. (\ref{eta_deacy}) and Eq. (\ref{global_gradient_bound}), it can be observed that, with the given mean-field terms, the convergence upper bound of the global loss function can be reduced by decreasing the learning rate over time. Moreover, according to Proposition \ref{pro_eta_t_0} and Eq. (\ref{w_gap_bound}), we have:
\begin{align} \label{eta_ineq}
\left\{\begin{array}{l}
\vspace{5pt}
-\frac{1}{\eta(t)}  \leq  -\frac{1}{ \sum_{i=1}^{N} \theta_{i} \frac{ 2\beta \theta_{i} D \| \boldsymbol{\phi_{1}}(t) \| }{(1 - \beta)\phi_{2}(t + 1)} }, \\
-\frac{1}{\eta(t)} \left(\sum_{i=1}^{N} \frac{2 \beta \theta_{i}^{2}  D^{2} \| \boldsymbol{\phi_{1}}(t) \|}{(1 - \beta)\phi_{2}(t+1)} \right) \leq - D. 
\end{array}\right.
\end{align}

Then, according to the expression of Eq. (\ref{w_gap_bound}) and Eq. (\ref{eta_ineq}), we can rewrite the Eq. (\ref{eq_label}) as follows:
\begin{align}
& F(\boldsymbol{w}(t \!+\! 1)) \!-\! F(\boldsymbol{w}(t)) \nonumber \\
& \leq \frac{L}{2} {\| \boldsymbol{w}(t \!+\! 1) \!-\! \boldsymbol{w}(t) \|}^{2} \!-\! \frac{1}{\eta(t)}\sum_{i=1}^{N} \frac{2 \beta \theta_{i}^{2} D^{2} \| \boldsymbol{\phi_{1}}(t) \|}{(1 \!-\! \beta)\phi_{2}(t \!+\! 1)} \nonumber \\
& \leq \frac{L}{2} {\| \boldsymbol{w}(t \!+\! 1) \!-\! \boldsymbol{w}(t) \|}^{2} \!-\! D \nonumber \\ 
& \leq \frac{L}{2} \left[\sum_{i=1}^{N} \frac{2 \beta \theta_{i}^{2} D^{2} \|\boldsymbol{\phi_{1}}(t)\|}{(1 \!-\! \beta)\phi_{2}(t \!+\! 1) }\right]^{2} \!-\! D.
\end{align}

Furthermore, for the classic FedAvg algorithm, the global model update rule is as follows:
\begin{align}
\textstyle \boldsymbol{w}(t + 1) = \boldsymbol{w}(t) - \eta \sum_{i=1}^{N} \theta_{i} \nabla F_{i} (\boldsymbol{w}(t)),
\end{align}
where $\eta > 0 $ is the static learning rate of the central server.  From the $L$-Lipschitz smoothness property of $F(\boldsymbol{w}(t))$ in Assumption \ref{Lipschitz_Smoothness} and Taylor expansion, we have:
\begin{align} \label{global_gradient_avg_bound}
& F(\boldsymbol{w}(t \!+\! 1)) \!-\! F(\boldsymbol{w}(t))  \nonumber \\
& \leq \left\langle \nabla F(\boldsymbol{w}(t)), \boldsymbol{w}(t \!+\! 1) \!-\! \boldsymbol{w}(t) \right\rangle  \!+\! \frac{L}{2} {\| \boldsymbol{w}(t \!+\! 1) \! - \! \boldsymbol{w}(t) \|}^2 \nonumber \\
& \leq - \frac{\| \boldsymbol{w}(t \!+\! 1) \!-\! \boldsymbol{w}(t) \| }{\eta} \!+\! \frac{L}{2} {\| \boldsymbol{w}(t \!+\! 1) \!-\! \boldsymbol{w}(t) \|}^2 \nonumber \\
& \leq \!\frac{L}{2} \!\! \left[\sum_{i=1}^{N} \! \frac{2 \beta \theta_{i}^{2}  D^{2} \|\boldsymbol{\phi_{1}}(t)\|}{(1 \!-\! \beta)\phi_{2}(t \!+\! 1) }\!\right]^{2} \!\!\!\!-\! \frac{\| \boldsymbol{w}(t \!+\! 1) \!-\! \boldsymbol{w}(t) \| }{\eta}. \!\! 
\end{align}

Based on Eq. (\ref{global_gradient_avg_bound}) and Assumption \ref{Bounded_Gradients}, we can obtain:
\begin{align} \label{D_comparison}
& - \frac{\| \boldsymbol{w}(t + 1) - \boldsymbol{w}(t) \| }{\eta} = -\frac{ \| \sum_{i=1}^{N} \theta_{i} \eta \nabla F_{i}(\boldsymbol{w}(t))  \| }{\eta} \nonumber \\
& \geq  \textstyle -\sum_{i=1}^{N} \theta_{i} \| \nabla F_{i}(\boldsymbol{w}(t)) \|  \geq  -\sum_{i=1}^{N} \theta_{i} D = -D. 
\end{align}

Thus, we have the following derivation:
\begin{align}\label{w_avg_lower}
\frac{L}{2} \left[\sum_{i=1}^{N}\!\! \right.& \left. \frac{2 \beta \theta_{i}^{2} D^{2}  \|\boldsymbol{\phi_{1}}(t)\|}{(1 \!-\! \beta)\phi_{2}(t \!+\! 1) }\right]^{2} \!\! - \! \frac{\| \boldsymbol{w}(t + 1) - \boldsymbol{w}(t) \| }{\eta} \nonumber \\ 
&\geq  \frac{L}{2} \left[\sum_{i=1}^{N} \frac{2 \beta \theta_{i}^{2} D^{2} \|\boldsymbol{\phi_{1}}(t)\|}{(1 \!-\! \beta)\phi_{2}(t \!+\! 1) }\right]^{2} \! - \! D. \!\!
\end{align}

From Eq. (\ref{w_avg_lower}), it can be easily observed that the convergence upper bound of our proposed algorithm FedEnt is smaller than that of the FedAvg algorithm.
\end{proof}

Based on the Theorem \ref{thm_F_global_bound}, it is shown that the convergence upper bound of the global loss function $F(\boldsymbol{w})$ between any two consecutive global training iterations in Eq. (\ref{lobal_gradient_bound_final}) is smaller than that of the FedAvg algorithm, which provides a trustworthy convergence upper bound for our method. 

\subsection{Bounding Convergence Rate}
According to the upper bound of the global loss function in Theorem \ref{thm_F_global_bound}, we conduct the convergence rate analysis of the FedEnt by measuring the difference between the global loss function $F(\boldsymbol{w}(t))$ at the $T$-th global training iteration and the optimal model $\boldsymbol{w^{*}}$ (similar to the prior research \cite{rakhlin2011making, stich2018local,spiridonoff2021communication}), which is summarized in following Theorem.

\begin{thm}\label{thm_convergence_rate}
Under the aforementioned Assumptions \ref{Bounded_Gradients}-\ref{Lipschitz_Smoothness} and Lemma~\ref{PL_inequality}, if $\kappa(t) \!=\! 1 - (1 + \frac{1}{1 - \beta})\frac{4 L \delta D^{2} \|\boldsymbol{\phi_{1}}\!(t)\|^{2}}{\phi^{2}_{2}(t+1)}$, the convergence rate of FedEnt after $T$ global training iterations satisfies,
\begin{align}\label{convergence_rate}
F(\boldsymbol{w}(T)) - F(\boldsymbol{w^{*}}) \leq & (F(\boldsymbol{w}(0)) - F(\boldsymbol{w^{*}})) \prod \limits_{t=0}^{T - 1} \kappa(t) \nonumber \\
& - D \left[1 + \sum_{t=1}^{T-1} \prod \limits_{t'=1}^{T-t} \kappa(t')\right],     
\end{align}
which can be faster than that of FedAvg as shown in Figs.~\ref{ex_mnist}-\ref{ex_cifar100}.
\end{thm}
\begin{proof}
Based on Theorem \ref{thm_F_global_bound}, we further analyze the upper bound of the global loss function as:
\begin{align} 
F(\boldsymbol{w}(t \!+\! 1)) &\!-\! F(\boldsymbol{w}(t)) \!\leq\! \underbrace{\frac{L}{2}\!\!\left[\sum_{i=1}^{N} \!\! \frac{2 \beta \theta_{i}^{2} \! D \|\boldsymbol{\phi_{1}}\!(t)\| \| \nabla \!F_{i} (\boldsymbol{w}(t)) \|}{(1 - \beta)\phi_{2}(t + 1) }\!\right]^{2}}_{\Omega} \!\!\!- D  \nonumber \\
&\Rightarrow \Omega \leq \frac{2 L \beta^{2}  D^{2}  \|\boldsymbol{\phi_{1}}(t)\|^{2}  (-\|\nabla \!F_{i} (\boldsymbol{w}(t))\|^{2})}{(\beta-1)^{3} \phi^{2}_{2}(t + 1)}. \nonumber
\end{align}

Then, as the inequality 
\begin{align}
\frac{\beta^{2}}{(\beta-1)^{3}} \!\leq\! \frac{\beta}{(\beta-1)^{3}} \!\leq\! \frac{\beta}{\beta-1} \!\leq\! 1 \!+\! \frac{1}{\beta-1} \!\leq\! 1 \!+\! \frac{1}{1-\beta},
\end{align}
holds and according to the Polyak-\L{}ojasiewicz inequality in Lemma \ref{PL_inequality}, Eq. (\ref{lobal_gradient_bound_final}) can be rewritten as follows:
\begin{align} \label{difference_with_optimal}
& F(\boldsymbol{w}(t\!+\!1)) \!-\! F(\boldsymbol{w^{*}})  \nonumber \\
&  \!\leq\! \left[\! 1 \!-\! (1 \!+\! \frac{1}{1 \!-\! \beta})\frac{4 L \delta  D^{2} \|\boldsymbol{\phi_{1}}\!(t)\|^{2}}{\phi^{2}_{2}(t + 1)}\! \right] \!\!(F(\boldsymbol{w}(t)) \!-\! F(\boldsymbol{w^{*}})) \!-\! D.
\end{align}

We define $ \kappa(t) \!=\! 1 - (1 + \frac{1}{1 - \beta})\frac{4 L \delta D^{2} \|\boldsymbol{\phi_{1}}\!(t)\|^{2}}{\phi^{2}_{2}(t+1)}$. Through iteratively adding both sides of the inequality in Eq. (\ref{difference_with_optimal}) at iteration $t \!\in\! \{0,1,\ldots, T-1\}$, we have:
\begin{equation*}
\resizebox{0.97\hsize}{!}{$\begin{aligned}
F(\boldsymbol{w}(T)) \!-\! F(\boldsymbol{w^{*}}) \!\leq\! (F(\boldsymbol{w}(0)) \!-\! F(\boldsymbol{w^{*}})) \prod \limits_{t=0}^{T - 1} \kappa(t) \!
- \! D \left[1 + \sum_{t=1}^{T-1} \prod \limits_{t'=1}^{T-t} \kappa(t')\right]. 
\end{aligned}$}   
\end{equation*}

Similar to the convergence analysis of FedEnt and define $S(t) \!=\! \frac{\| \boldsymbol{w}(t + 1) - \boldsymbol{w}(t) \|}{\eta}$, the convergence rate of FedAvg is derived as follows:
\begin{align}
F(\boldsymbol{w}(T)) \!-\! F(\boldsymbol{w^{*}}) &\leq  (F(\boldsymbol{w}(0)) \!-\! F(\boldsymbol{w^{*}})) \prod \limits_{t=0}^{T - 1} \kappa(t)  \nonumber \\
& - \! \left[S(T \!-\! 1) + \sum_{t=0}^{T-2} S(t) \!\! \prod \limits_{t'=1}^{T-t-1} \!\!\! \kappa(t') \!\right].  
\end{align}

\begin{figure*}[t]
\setlength{\abovecaptionskip}{2pt} 
    \centering
    \subfloat[IID ($\alpha_{d} \! \rightarrow \! +\infty$)]{
        \label{Cifar100_IID}
        \includegraphics[width=0.33\textwidth, trim=55 16 95 45,clip]{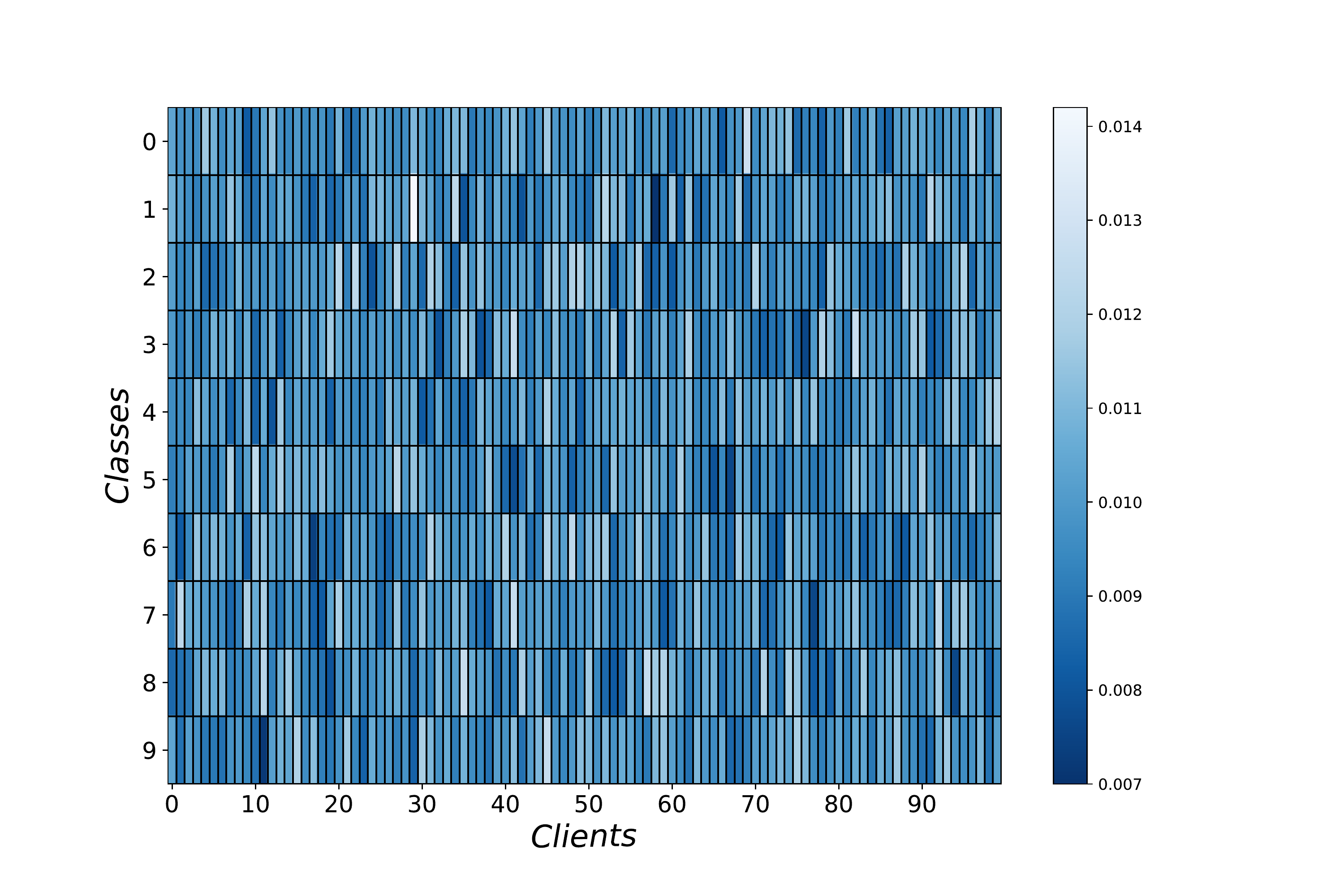}}
    \subfloat[$\alpha_{d} = 1.0$]{
        \label{Cifar100_Non_IID_1}
        \includegraphics[width=0.33\textwidth, trim=55 16 95 45,clip]{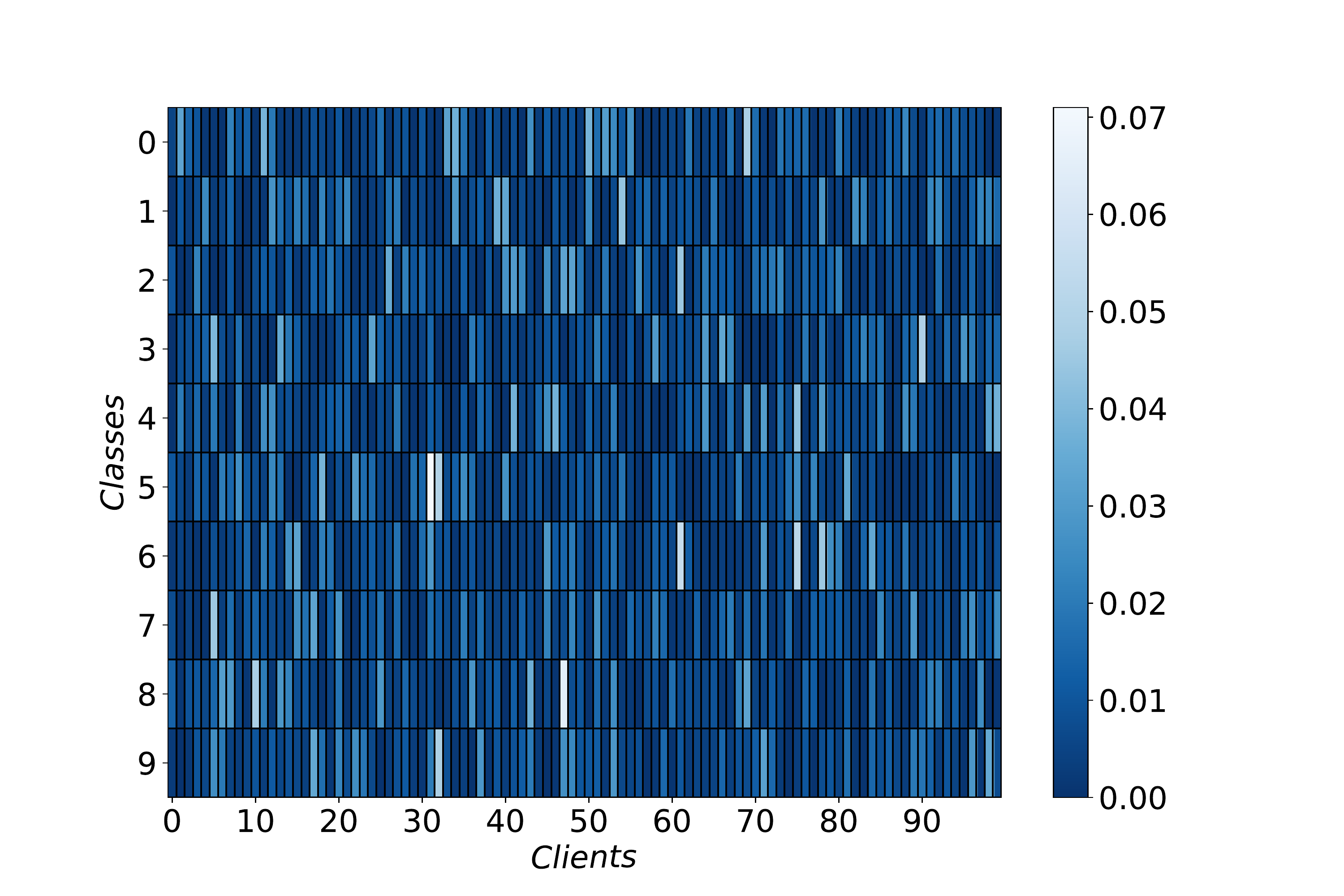}}
    \subfloat[$\alpha_{d} = 0.5$]{
        \label{Cifar100_Non_IID_0_5}
        \includegraphics[width=0.33\textwidth, trim=55 16 95 45,clip]{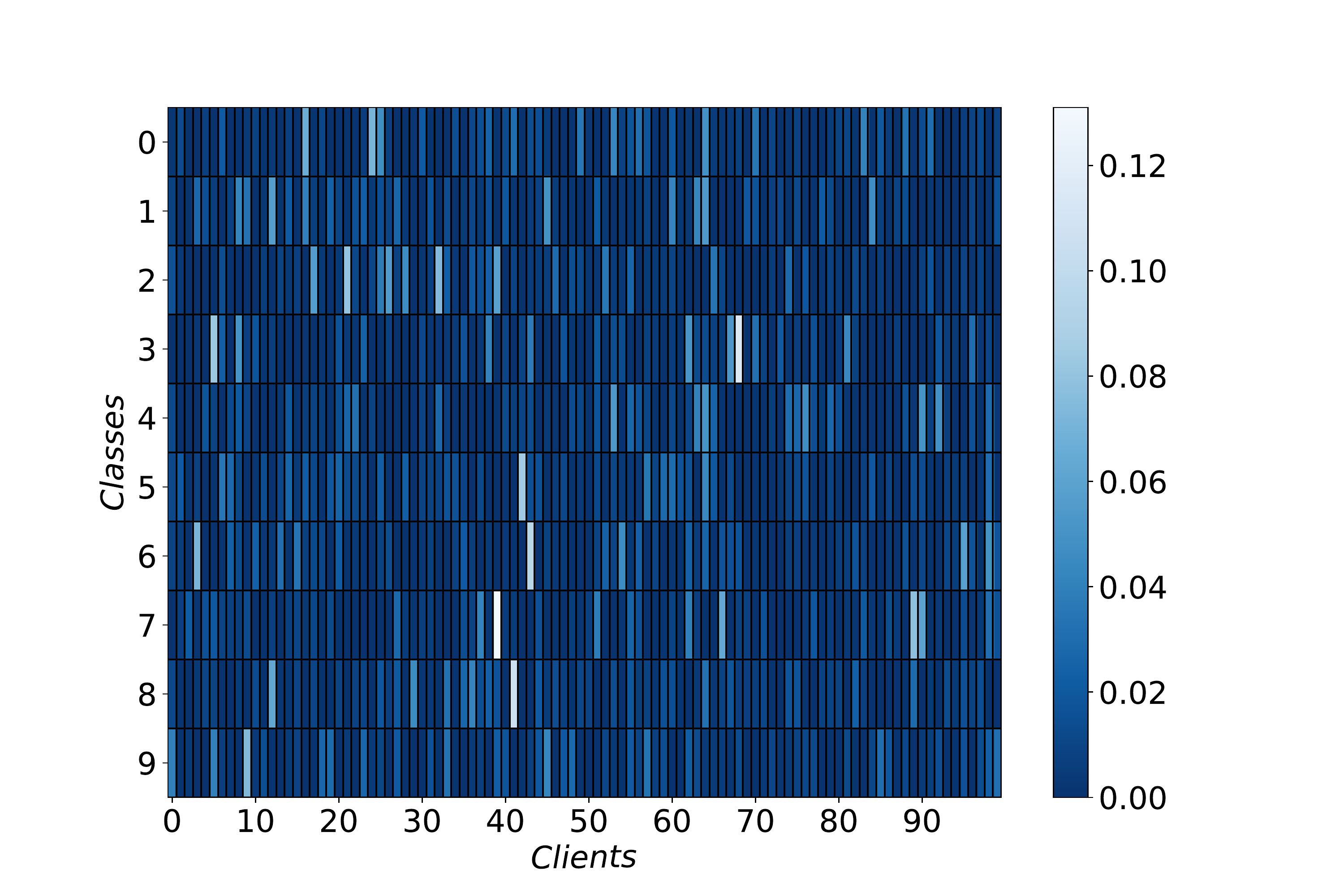}} 
    \caption{The different data distributions on the CIFAR10 dataset.}
    \label{Dirichlet_data_distribution}
\vspace{-8pt}
\end{figure*}

\begin{table*}[t]
\centering
\setlength{\abovecaptionskip}{2pt} 
\renewcommand\arraystretch{1.0}
\caption{Accuracy on Different Datasets with 20\% of 100 clients.}
\resizebox{1.0\textwidth}{!}{\begin{tabular}
{c|c|c@{\hspace{8.2pt}}|c@{\hspace{8.2pt}}|c@{\hspace{8.2pt}}|c@{\hspace{8.2pt}}|c@{\hspace{8.2pt}}} 
\toprule[1.2pt]
\textbf{Dataset} & \textbf{Distribution} & \parbox{2cm}{\centering\textbf{FedAvg}\cite{r1}} & \parbox{2cm}{\centering\textbf{FedAdam}\cite{reddi2020adaptive}} & \parbox{2cm}{\centering\textbf{FedProx}\cite{li2020federated}} & \parbox{2cm}{\centering\textbf{FedDyn}\cite{acar2021federated}} & \parbox{2cm}{\centering\textbf{FedEnt (ours)}} \\ \midrule[0.8pt]

\parbox{2cm}{\centering\textbf{MNIST}\\\centering\textbf{+MNIST-CNN}}
        & Pathological Non-IID & 85.53  & 94.89  & 86.72  & 88.95 &  \textbf{97.24}  \\ \midrule[0.8pt]

\multirow{2}{*}{\multirowcell{2}{\textbf{CIFAR10} \\ \textbf{+CIFAR10-CNN}}}
        
        & $\alpha_{d} = 1.0$ & 56.41  & 66.89  & 67.87  & 67.11  & \textbf{69.02}  \\ \cmidrule[0.5pt](l{1pt}r{0pt}){2-7}
        
        & $\alpha_{d} = 0.5$ & 55.52  & 63.21  & 65.91  & 66.18  & \textbf{67.13}   \\ \midrule[0.8pt]

\multirow{2}{*}{\multirowcell{2}{\textbf{EMNIST-L} \\ \textbf{+LeNet-5}}}
        
        & $\alpha_{d} = 1.0$ & 91.89  & 92.18  & 92.01  & 92.37  & \textbf{92.92}   \\ \cmidrule[0.5pt](l{1pt}r{0pt}){2-7}
        
        & $\alpha_{d} = 0.5$ & 90.80 & 91.10  & 90.90  & 90.93  & \textbf{92.15} \\ \midrule[0.8pt]

\multirow{2}{*}{\multirowcell{2}{\textbf{CIFAR100} \\ \textbf{+VGG-11}}}
        & $\alpha_{d} = 1.0$ & 44.43  & 49.75  & 44.18  & 44.45  & \textbf{52.25}   \\ \cmidrule[0.5pt](l{1pt}r{0pt}){2-7}
        & $\alpha_{d} = 0.5$ & 43.91  & 49.90  &  43.96 & 44.84  & \textbf{51.15}
        \\ 
\bottomrule[1.2pt]
\end{tabular}}
\label{tab_experient_results}
\vspace{-12pt}
\end{table*}

\begin{figure}[t]
\setlength{\abovecaptionskip}{1pt} 
\centerline{\includegraphics[width=0.5\textwidth]{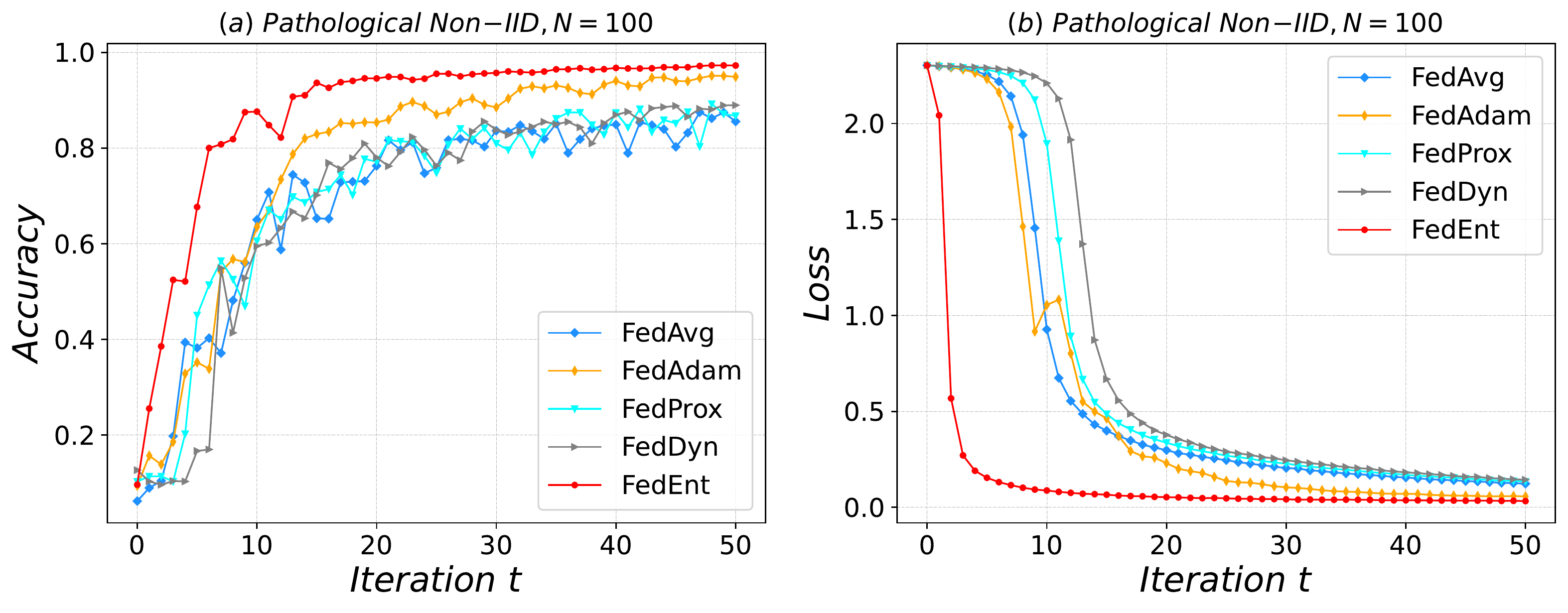}}
\caption{The performance of different FL methods on the MNIST dataset with 20\% of 100 clients.}
\label{ex_mnist}
\vspace{-15pt}
\end{figure}

According to the inequality in Eq. (\ref{D_comparison}), we have:
\begin{align} \label{convergence_rate_comparioson}
& \small (F(\boldsymbol{w}(0)) \!-\! F(\boldsymbol{w^{*}}))  \prod \limits_{t=0}^{T - 1} \kappa(t)\! 
- \! \left[S(T \!-\! 1) + \sum_{t=0}^{T-2} S(t) \!\! \prod \limits_{t'=1}^{T-t-1} \!\!\! \kappa(t') \!\right] \nonumber \\
& \small \!\geq\! (F(\boldsymbol{w}(0)) \!-\! F(\boldsymbol{w^{*}})) \!\prod \limits_{t=0}^{T - 1} \!\! \kappa(t)  \!-\! D \left[1 + \sum_{t=1}^{T-1} \! \prod \limits_{t'=1}^{T-t}\! \kappa(t')\right].\!\!
\end{align}

From Eq. (\ref{convergence_rate_comparioson}), we can observe that the convergence bound of our proposed FedEnt is smaller than that of FedAvg, which implies that our proposed adaptive algorithm FedEnt could achieve a faster convergence rate. Hence, the theorem holds.
\end{proof}

\begin{figure*}[t]
\setlength{\abovecaptionskip}{-1pt} 
\centerline{\includegraphics[width=0.84\textwidth]{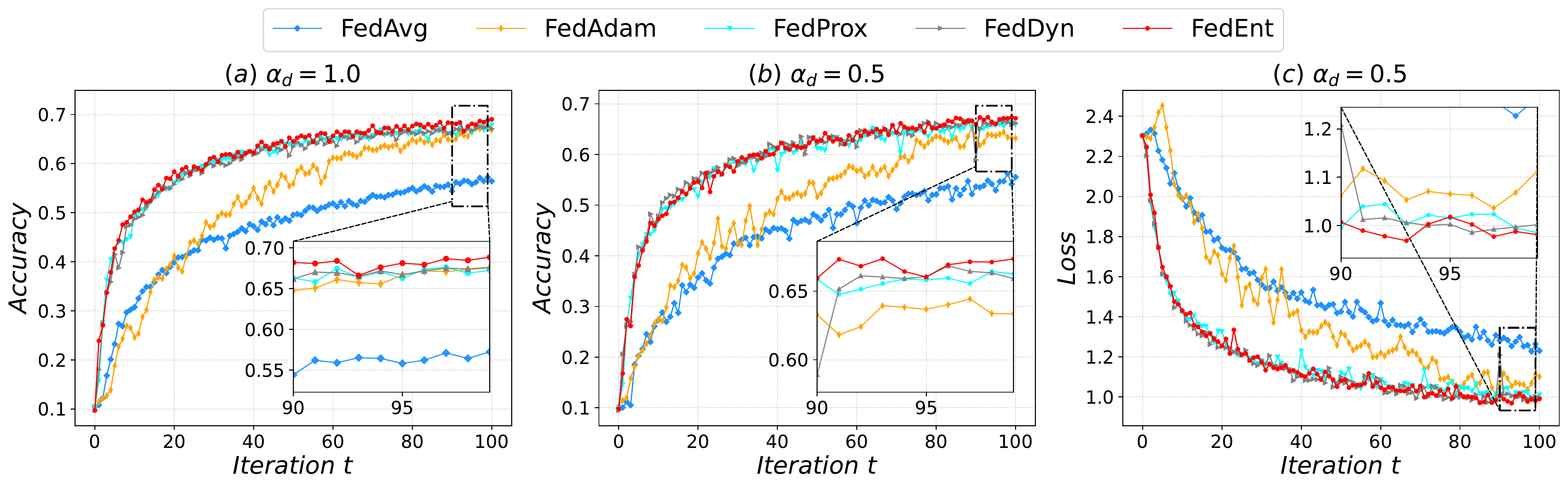}}
\caption{The performance of different FL methods on the CIFAR10 dataset with 20\% of 100 clients.}
\label{ex_cifar10}
\vspace{-14pt}
\end{figure*}

\begin{figure*}[t]
\setlength{\abovecaptionskip}{-1pt} 
\centerline{\includegraphics[width=0.84\textwidth]{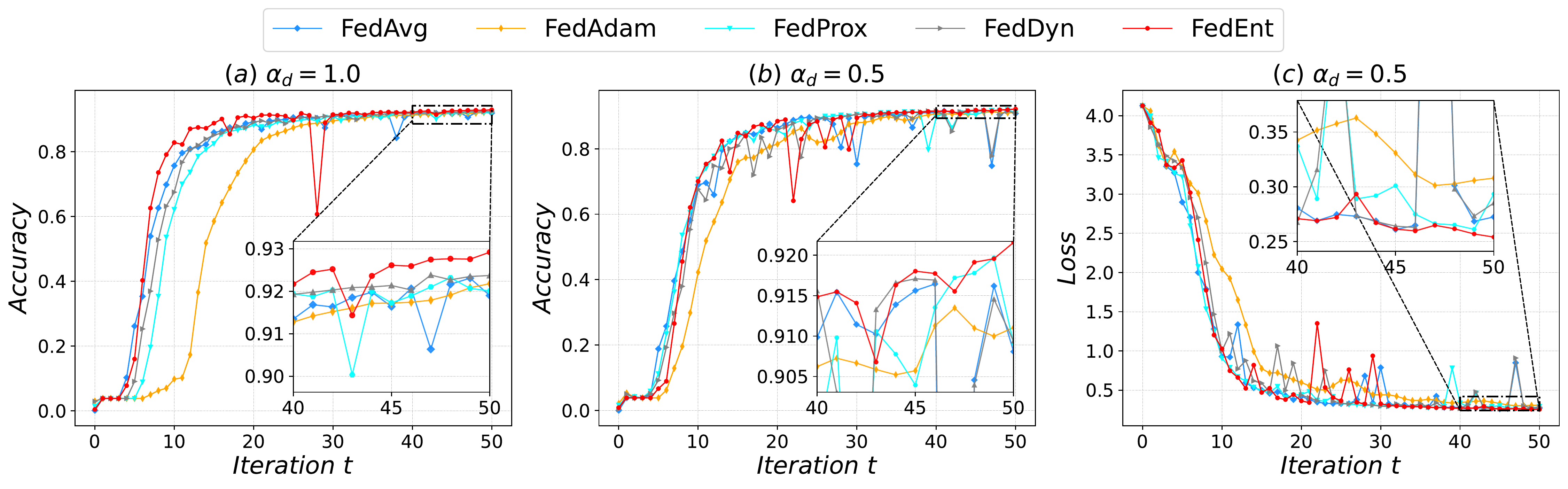}}
\caption{The performance of different FL methods on the EMNIST-L dataset with 20\% of 100 clients.}
\label{ex_emnist}
\vspace{-16pt}
\end{figure*}

Theorem \ref{thm_convergence_rate} provides an important insight into the convergence property of the FedEnt algorithm. It facilitates the selection of hyperparameters to enhance the convergence rate of FedEnt, suggesting that a higher value of the aggregation weight $\beta$ may result in faster convergence of the FL mode.

\section{Experiments}\label{sec_experiments}
In this section, we conduct extensive experiments on the MNIST, CIFAR10, EMNIST-L, and CIFAR100 datasets to verify the performance of our proposed adaptive learning FL algorithm FedEnt. We first introduce the experiment setup, including the dataset, local training model, and baseline FL methods. Then, we compare the performance of our adaptive learning rate method with other FL algorithms in terms of data distributions, client participation ratio, distance metric, parameter deviation, and hyperparameter efficiency. Experimental results show that our method outperforms the benchmarks.

\begin{table}[t]
\setlength{\abovecaptionskip}{-4pt} 
\caption{Hyperparameter of all Datasets}
\renewcommand\arraystretch{1.0}
\begin{center}
\resizebox{0.49\textwidth}{!}{\begin{tabular}{c|c|c|c|c}
\toprule[1pt]
\parbox{3.0cm}{\centering\textbf{Datasets+Local Model}}&\parbox{1cm}{\centering\textbf{Learning}\\\centering\textbf{Rate}}&\textbf{Batchsize}&\parbox{1cm}{\centering\textbf{Global}\\\centering\textbf{Iteration}} &\parbox{1cm}{\centering\textbf{Local}\\\centering\textbf{Epoch}}  \\ \cmidrule[0.5pt](l{1pt}r{0pt}){1-5}

\parbox{3.0cm}{\centering\textbf{MNIST+MNIST-CNN}} & 0.01 & 32  & 50 & 3 \\ \cmidrule[0.5pt](l{1pt}r{0pt}){1-5}

\parbox{3.0cm}{\centering\textbf{EMNIST-L+LeNet-5}} & 0.1 & 64  & 50 & 3 \\ \cmidrule[0.5pt](l{1pt}r{0pt}){1-5}

\parbox{3.3cm}{\centering\textbf{CIFAR10+CIFAR10-CNN}}& 0.1 & 20  & 100 & 5 \\ \cmidrule[0.5pt](l{1pt}r{0pt}){1-5}

\parbox{3cm}{\centering\textbf{CIFAR100+VGG-11}}& 0.01 & 32  & 110 & 3 \\ 
\bottomrule[1pt]
\end{tabular}}
\label{hyperparameter}
\end{center}
\vspace{-16pt}
\end{table}

\subsection{Experiment Setup}

\subsubsection{Dataset and Data Distribution}

\textbf{Dataset:} We evaluate the algorithm on four real-world datasets (MNIST \cite{lecun1998gradient}, CIFAR10 \cite{krizhevsky2009learning}, EMNIST-L \cite{cohen2017emnist}, and CIFAR100 \cite{krizhevsky2009learning}). MNIST dataset contains 60,000 training data images and 10,000 test set images, each of which is a $28 \times 28$ pixel grayscale image, representing a handwritten digit from 0 to 9. The CIFAR10 dataset consists of 60,000 $32 \times 32$ sized RGB images of 10 categories, 50,000 of which serve as the training set and 10,000 as the validation set. EMNIST-L dataset comprises a total of 103,600 grayscale images from 37 categories, each of which contains an equal amount of data, including 2400 images as the training set and 400 images as the test set. There are 100 classes in the CIFAR100 dataset, and each class has 600 RGB images with size $32 \times 32$, where for each category, 500 images are set for model training and 100 images are set for performance test. 

\noindent \textbf{Data Distribution:} For the MNIST dataset, we partition it into Non-IID with a Pathological distribution similar to \cite{r1}. For the rest of the datasets, we generate Non-IID data based on the Dirichlet distribution \cite{hsu2019measuring}. We assume that each participating client's training example is drawn independently with class labels following a categorical distribution over $M$ classes parameterized by a vector $\boldsymbol{q} (q_{i} \! > \! 0, i \! \in \! \{1, 2, \ldots, M\}$ and $| \boldsymbol{q} |_{1} \! = \! 1)$. Then, we generate $\boldsymbol{q} \! \sim \! Dir(\alpha_{d} \boldsymbol{q})$ from a Dirichlet distribution, where $\boldsymbol{q}$ characterizes a prior class distribution over $M$ classes, and $\alpha_{d} \! > \! 0$ is a concentration parameter controlling the identicalness among clients. With $\alpha_{d} \! \rightarrow \! +\infty$, all clients have identical distributions to the prior; while with $\alpha_{d} \! \rightarrow \! 0$, on the other extreme, each client holds examples from only one class chosen at random \cite{hsu2019measuring}. The data distributions on the CIFAR10 dataset under IID distribution ($\alpha_{d} \!\! \rightarrow \!\! +\infty$) and Non-IID distribution with different Dirichlet parameters $\alpha_{d}$ are displayed in Fig. \ref{Dirichlet_data_distribution}. With the increment of $\alpha_{d}$, the color of the heat map that depicts the training label distribution becomes deeper as illustrated in 
Fig. \ref{Dirichlet_data_distribution}(b)-\ref{Dirichlet_data_distribution}(c). In the meantime, the range of label ratio is wider, which indicates a higher level of data heterogeneity.

\begin{figure*}[t]
\setlength{\abovecaptionskip}{0pt} 
\centerline{\includegraphics[width=0.85\textwidth]{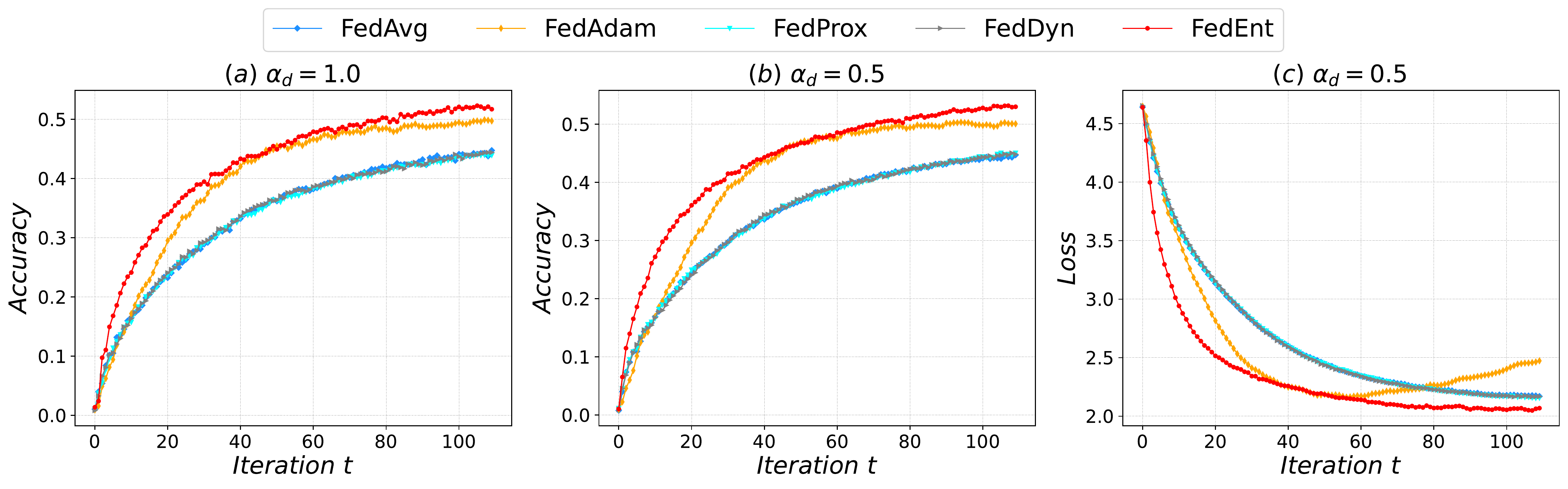}}
\caption{The performance of different FL methods on the CIFAR100 dataset with 20\% of 100 clients.}
\label{ex_cifar100}
\vspace{-14pt}
\end{figure*}

\begin{figure*}[t]
\setlength{\abovecaptionskip}{1pt} 
\centerline{\includegraphics[width=0.85\textwidth]{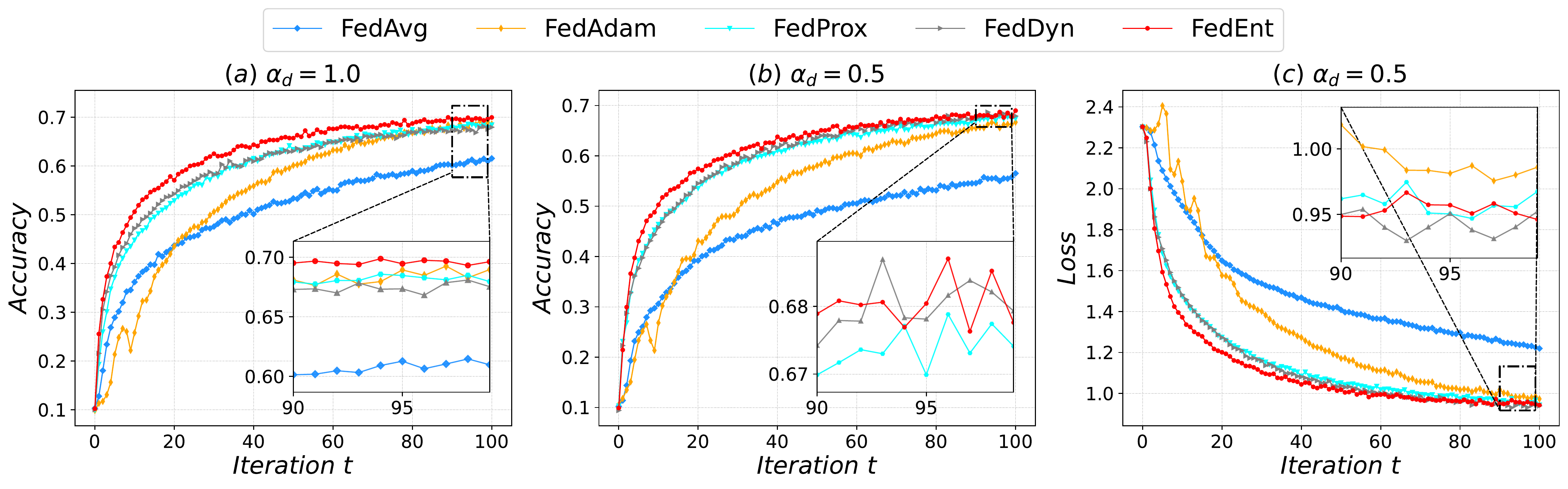}}
\caption{The performance of different FL methods on the CIFAR10 dataset with 100\% of 100 clients. }
\label{cifar10_100*100}
\vspace{-15pt}
\end{figure*}

\subsubsection{Model}
We utilize four different local training models in our experiments, each of which is specific to the corresponding dataset. For the MNIST dataset, we use an MNIST-CNN model, which contains two $5 \! \times \! 5$ convolution layers, each convolution layer has a layer followed by $2 \! \times \! 2$ max pooling, two dropout layers. Each dropout layer is followed by a fully connected layer of $7 \! \times \! 7 \! \times \! 64$ and 512 units, respectively, and a ReLU output layer of 10 units. For the CIFAR10 dataset, a CNN model is implemented on CIFAR10, which includes three $3 \! \times \! 3$ convolution layers, each convolution layer has a layer followed by $2 \! \times \! 2$ max pooling, two dropout layers. Each dropout layer is followed by a fully connected layer of $4 \! \times \! 4 \! \times \! 64$ and 512 units, respectively, and a ReLU output layer of 10 units. For the EMNIST-L and CIFAR100 datasets, we adopt the LetNet-5 \cite{lecun1998gradient} and VGG-11 \cite{simonyan2014very} respectively.

\subsubsection{Baseline and Parameters Setting}
We compare our proposed adaptive learning rate method FedEnt with other state-of-the-art FL algorithms, including FedAvg \cite{r1}, FedAdam \cite{reddi2020adaptive}, FedProx \cite{li2020federated}, and FedDyn \cite{acar2021federated}. To more comprehensively evaluate the performance of our FedEnt in terms of measuring the difference between the global model and local models, we have also introduced the following two FL algorithms as comparisons:
\begin{itemize}[itemsep=0pt, leftmargin=*, align=right]
    \item FedCos: Refer to \cite{wu2021fast}, this algorithm uses cosine similarity to measure the parameter difference between the global and local models. Specifically, when updating the global model parameters, FedCos assigns update weights based on the cosine similarity between the local and global models. In FedCos, the cosine similarity is calculated using the following formula:
    \begin{align}\label{cos_objective}
    \small \omega_{i}(t) = \arccos \left(\frac{\langle \boldsymbol{w_{i}}(t), \boldsymbol{w}(t) \rangle}{\Vert \boldsymbol{w_{i}}(t) \Vert \Vert \boldsymbol{w}(t) \Vert}\right), \theta_{i}(t) = \cos (\omega_{i}(t)),
    \end{align}
    where $\omega_{i}(t)$ represents the cosine similarity between client $i \!\in\! \{1, 2, \ldots, N\}$ and the global model at iteration $t \!\in\! \{1, 2, \ldots, T\}$. By computing the cosine value $\theta_{i}(t) = \cos(\omega_{i}(t))$, FedCos converts the similarity into weights for updating the global model parameters. This similarity metric helps capture directional differences between model parameters.
    \item FedNorm \cite{tu2022adaptive} utilizes the Euclidean norm to measure the parameter difference between the global and local models. In this case, the adaptive objective optimization function for client $i \in \{1, 2, \ldots, N\}$ is transformed into:
    \begin{align}\label{norm_objective}
        \small U_{i}(T)  \!= & \small \min_{\substack{\eta_{i}(t)}}  \sum_{t=0}^{T} \!
        \Bigg( \! (1 \!-\! \beta)\eta_{i}^{2}(t) \!+\! \beta \sum_{j=1}^{N} {\Vert \boldsymbol{w_{i}}(t) \!-\! \boldsymbol{w}(t)\Vert}^{2} \! \Bigg) , \\ 
     \text{s.t.} \quad\quad & \small \boldsymbol{w_i}(t+1) = \boldsymbol{w}(t) - \eta_i(t)\nabla F_i(\boldsymbol{w}(t)), \nonumber
    \end{align}
    where $\eta_{i}(t)$ represents the learning rate of client $i$ at iteration $t$, and $\beta$ is a hyperparameter that balances the impact of model differences and learning rates. FedNorm aims to optimize the training effectiveness by minimizing the Euclidean distance between the global and local models.
\end{itemize}

\begin{figure*}[t]
\setlength{\abovecaptionskip}{-1pt} 
\centerline{\includegraphics[width=0.84\textwidth]{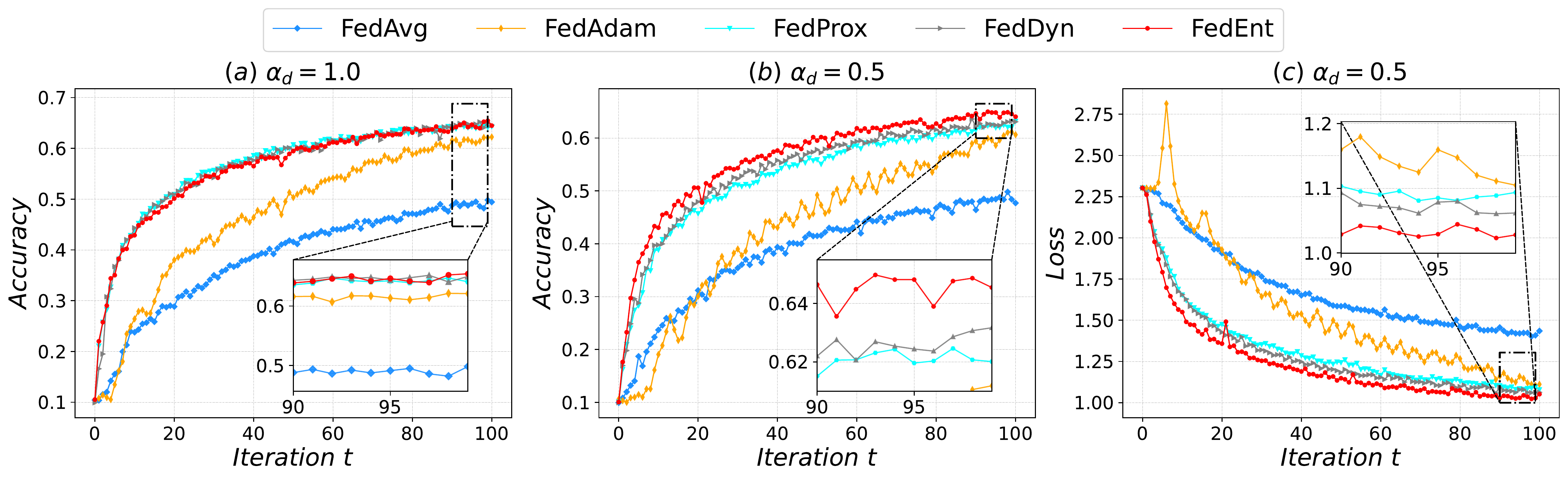}}
\caption{The performance of different FL methods on the CIFAR10 dataset with 20\% of 200 clients. }
\label{cifar10_20*200}
\vspace{-15pt}
\end{figure*}

Refer to \cite{reddi2020adaptive}, we adopt similar hyperparameter settings with $\beta_{1} \! = \! 0.9$, $\beta_{2} \! = \! 0.99$ and $\tau \! = \! 10^{-3}$ for FedAdam. Regarding FedProx and FedDyn, their key hyperparameters are set to $\mu \! = \! 0.01$ and $\alpha \! = \! 0.001$, respectively. For FedNorm and our method, we set weight $\beta \! = \! 0.99$ and weight decay $\gamma \! = \! 0.99$. We consider $N \! = \! 100$ clients during the training process and default client sample rate as 20\%. The rest of the hyperparameter settings are summarized in Table \ref{hyperparameter}.

\begin{figure*}[t]
\setlength{\abovecaptionskip}{-1pt} 
\centerline{\includegraphics[width=0.84\textwidth]{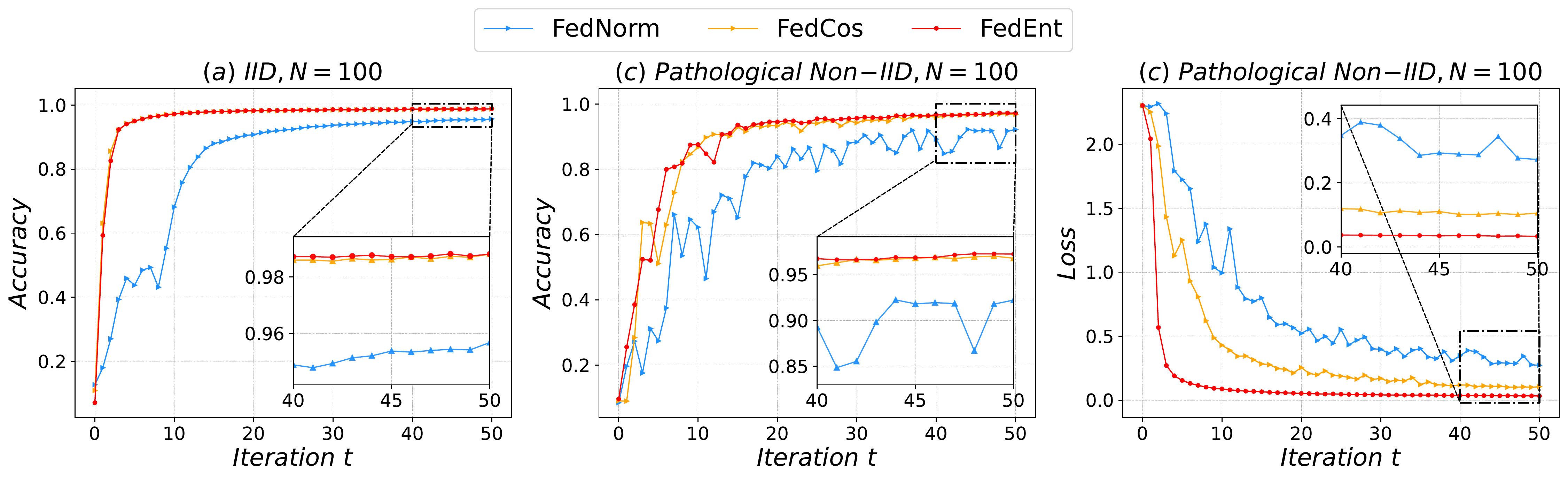}}
\caption{The performance of different distance metric FL methods on the MNIST dataset with 20\% of 100 clients.}
\label{ex_mnist_similarity}
\vspace{-16pt}
\end{figure*}

\begin{figure*}[t]
\setlength{\abovecaptionskip}{-1pt} 
\centerline{\includegraphics[width=0.84\textwidth]{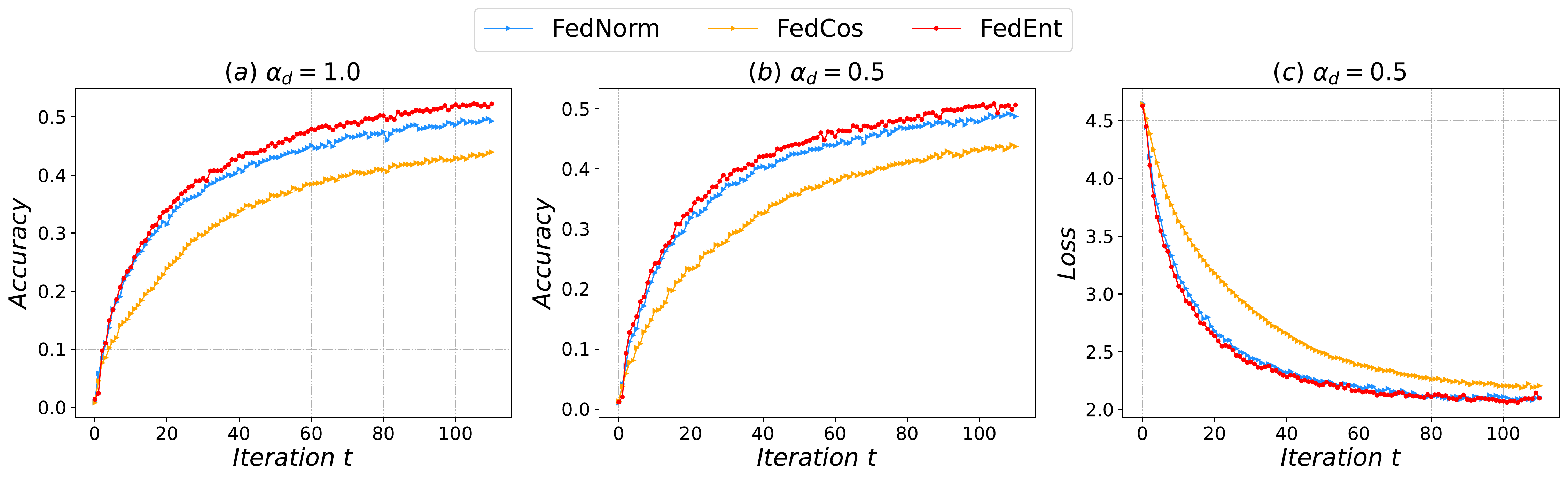}}
\caption{The performance of different distance metric FL methods on the CIFAR10 dataset with 20\% of 100 clients. }
\label{ex_cifar10_similarity}
\vspace{-16pt}
\end{figure*}

\subsection{Numerical Results}
We demonstrate the superiority of FedEnt from five different aspects, including data heterogeneous analysis on model performance under various data distributions, participation ratio analysis on the global model accuracy, distance metric analysis that evaluates the differences between global and local models, deviation analysis on the local model performance, and effectiveness of hyperparameters in model training.

\subsubsection{Data Heterogeneous Analysis}
We evaluate the effectiveness of the proposed FedEnt algorithm by comparing the accuracy of different FL algorithms on different datasets. The experimental results on Non-IID data are summarized in Table \ref{tab_experient_results}, and our proposed FedEnt algorithm achieves the best performance in all dataset cases. Specifically, FedEnt surpasses FedAvg with an 11\% increase in accuracy, while FedAdam shows a 4\% accuracy boost, FedProx exhibits a 9\% increase in accuracy, and FedDyn demonstrates a 7\% accuracy improvement. Additionally, we evaluate the convergence rate of FedEnt by comparing the accuracy and training loss curve of the learned global model under different data heterogeneity as displayed in Figs. \ref{ex_mnist}-\ref{ex_cifar100}. Taking the results on the MNIST dataset as an example, we can observe an obvious convergence rate gap between FedEnt and other baseline algorithms under the Non-IID distribution as shown in Fig. \ref{ex_mnist}(a). Meanwhile, in Fig. \ref{ex_mnist}(b), by assigning an adaptive learning rate in each global iteration, FedEnt effectively decreases the training loss during model training. The performance on the EMNIST-L dataset with different concentration parameters $\alpha_{d}$ in Fig. \ref{ex_emnist} is similar to that on MNIST due to their similar data features. Furthermore, despite the significantly higher task complexity of the CIFAR100 dataset, the enhancement with diverse data heterogeneity is still notable and effectively avoids over-fitting during model training as illustrated in Fig. \ref{ex_cifar100}(c), which is consistent with the analytical results in Theorem \ref{thm_F_global_bound}. FedEnt converges rapidly during the initial training stage thanks to the adaptive learning rate scheme and prominently improves the model performance.

\subsubsection{Participation Ratio Analysis}
We investigate how the number of clients affects the model performance of the learned FL global model by adjusting the participation ratio and the total number of participants. Firstly, we conduct simulations with 20\% of 100 mobile clients as shown in Figs. \ref{ex_mnist}-\ref{ex_cifar100}. The experimental results show that our proposed adaptive algorithm FedEnt outperforms the baselines with a higher accuracy and a faster convergence rate. Then, we modify the participation ratio of the clients to 100\% of 100 clients on the CIFAR10 dataset as displayed in Fig. \ref{cifar10_100*100}. It is observed that FedEnt still achieves the better  model performance in comparison to other state-of-the-art FL algorithms. In addition, with the increment of the participants, FedEnt has more prominent advantages over other baselines on model accuracy and the convergence rate. In Fig. \ref{cifar10_20*200}, we further increase the number of the total mobile clients to 200, and our proposed method still performs the best, which indicates the robustness and effectiveness of FedEnt under both small-scale and large-scale heterogeneous participants.

\subsubsection{Distance metric analysis between global and local models}

To compare the model performances of various distance metrics that evaluate the differences between global and local models (such as cosine similarity, norm forms, and entropy approach) in FL, the simulation experiments are conducted on the MNIST and CIFAR10 datasets. Firstly, the experimental results presented in Fig. \ref{ex_mnist_similarity} on the simple MNIST dataset reveal that the FedEnt algorithm still exhibits superior performance under both IID and Pathological Non-IID data distributions. This indicates its effectiveness in bridging the gap between the global and local models. In comparison, FedCos in Eq. (\ref{cos_objective}) which utilizes the cosine similarity between global and local models to adjust the FL model performs slightly inferior to FedEnt, demonstrating a degree of robustness. However, FedNorm in Eq. (\ref{norm_objective}) which utilizes the norm forms to measure the difference between the global and local models exhibits poorer performance across both data distributions, particularly under Pathological Non-IID.

The superiority of our proposed entropy-based method is more pronounced with the CIFAR10 dataset as shown in Fig. \ref{ex_cifar10_similarity}. It is evident that as the value of $\alpha_d$ increases (indicating a higher degree of Non-IID), the performance of all metrics declines,  which signifies the negative impact of the Non-IID nature of data on model performance. Nevertheless, regardless of the varying degrees of Non-IID, FedEnt consistently maintains a higher accuracy, demonstrating its resilience across different Non-IID scenarios. In conclusion,  FedEnt based on the entropy approach possesses a significant advantage compared with other distance metrics,  as it demonstrates consistently higher performance across various datasets and varying levels of Non-IID degrees. 

\begin{figure*}[t]
\setlength{\abovecaptionskip}{2pt} 
    \centering
    \begin{minipage}{170pt}
    \centerline{\includegraphics[width=1.0\textwidth, trim=8 5 5 7,clip]{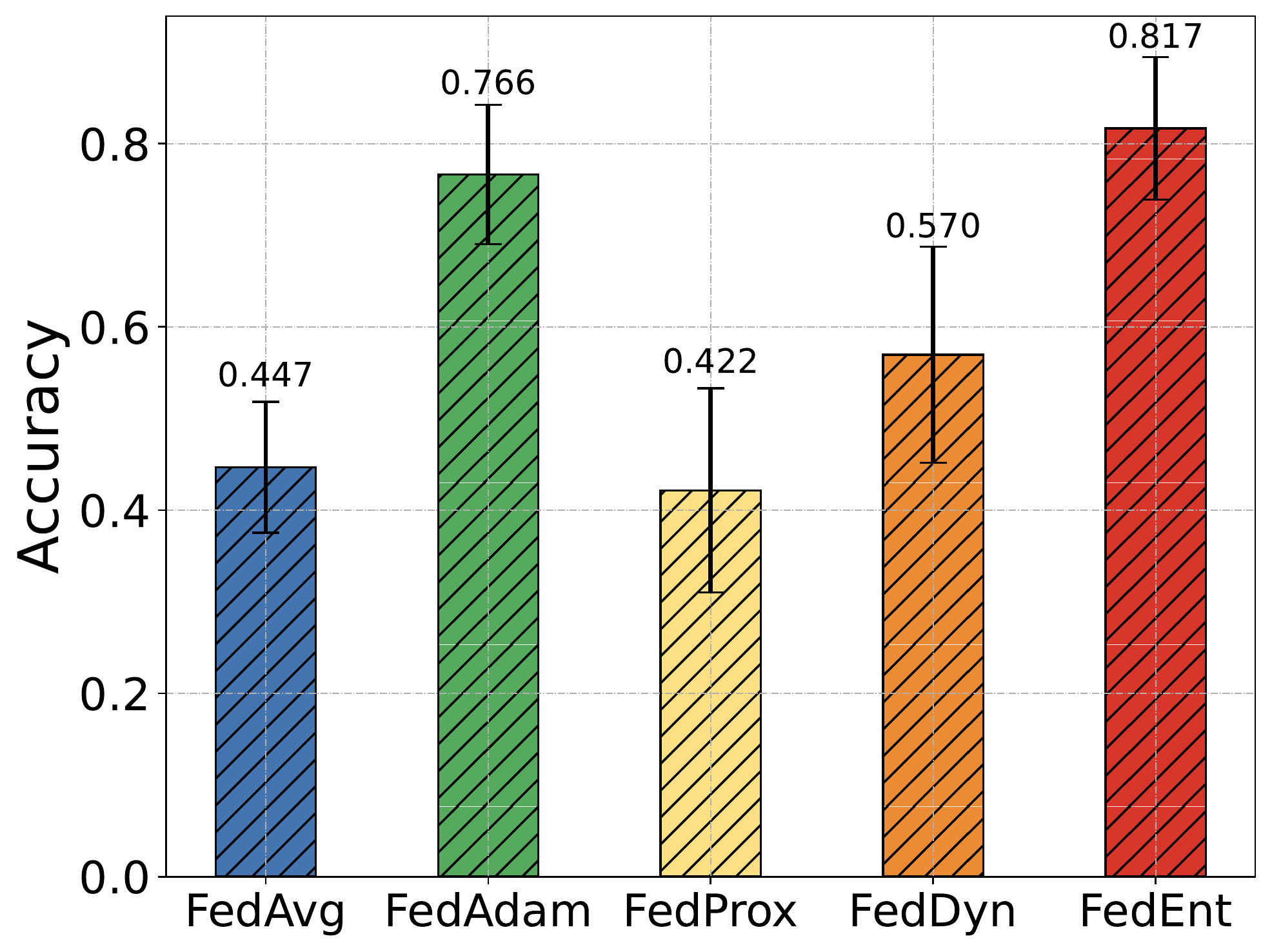}}
        \caption{The accuracy of local model parameters on the MNIST test dataset.}
        \label{client_dev}
    \end{minipage}
    \hspace{2pt}
    \begin{minipage}{320pt}
        \subfloat[Accuracy Curve]{
            \label{MNIST_beta_acc}
            \includegraphics[width=0.5\textwidth, trim=40 10 80 65,clip]{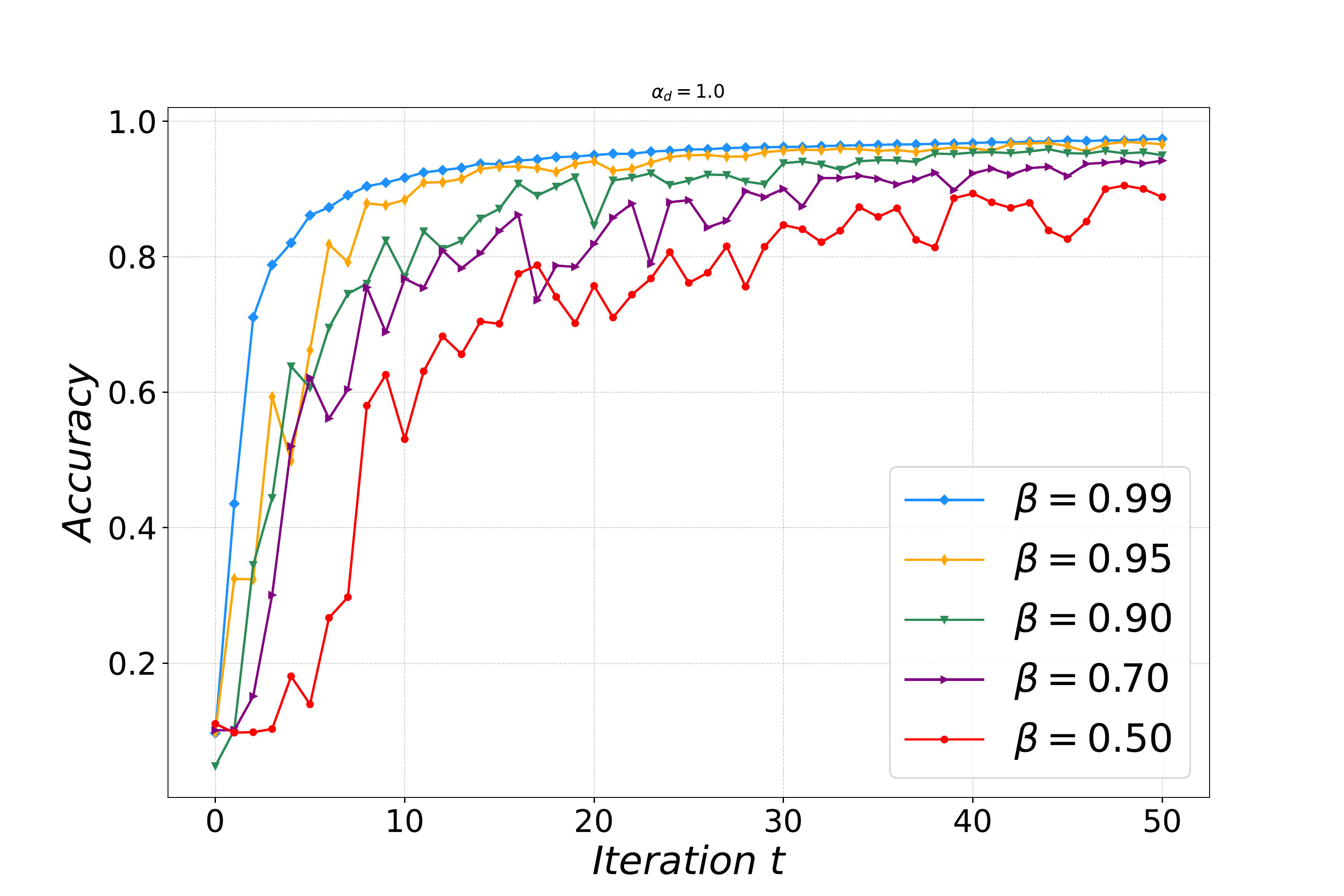}}
        \subfloat[Loss Curve]{
            \label{MNIST_beta_loss}
            \includegraphics[width=0.5\textwidth, trim=40 10 80 65,clip]{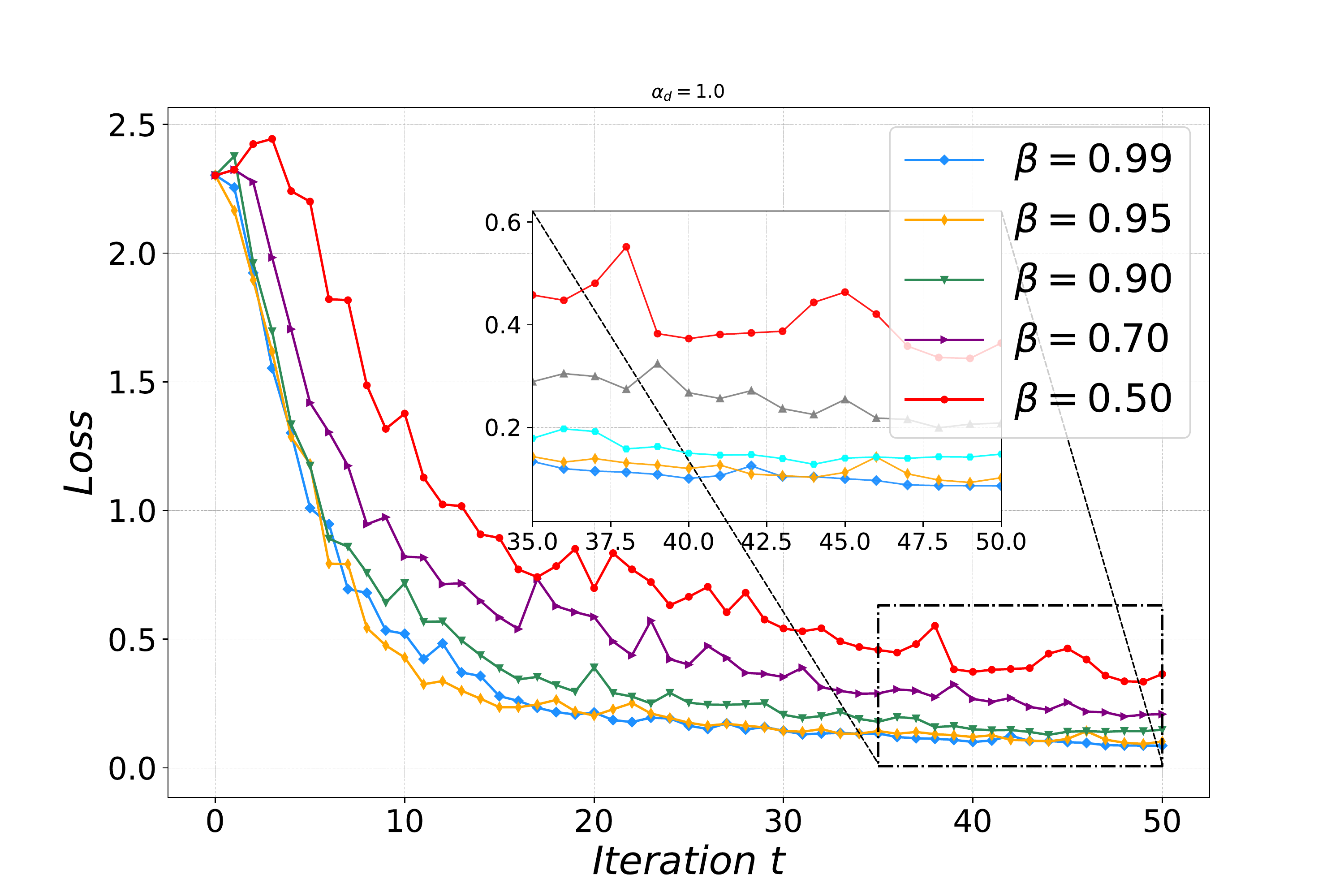}} 
         \caption{The impact of the weight aggregation parameter $\beta$ on the performance of FedEnt in the MNIST dataset.}
        \label{beta}
    \end{minipage}
\vspace{-18pt}
\end{figure*}

\subsubsection{Deviation Analysis}
We conduct deviation analysis to demonstrate the effectiveness of our proposed FedEnt in eliminating client drifting compared with other state-of-the-art FL algorithms, by analyzing the accuracy of the local parameter on the global test dataset during model training. Under the Pathological Non-IID data distribution, the data distribution on local clients varies considerably compared to that of the global test dataset on the central server. In Fig. \ref{client_dev}, we display the performance of different FL algorithms on the MNIST dataset with 20\% of participating clients. Our proposed adaptive FL algorithm FedEnt achieves the highest average accuracy and lowest accuracy error interval, which validates the effectiveness of the entropy term in mitigating the negative influence of client drifting and improving model performance in the federated network, as shown in Proposition \ref{thm_w_global_local_bound}. 

\begin{figure}[t]
\setlength{\abovecaptionskip}{0.5pt}
\begin{minipage}{245pt}
\subfloat[Accuracy Curve]{
    \label{MNIST_gamma_acc}
    \includegraphics[width=0.5\textwidth,  trim=40 10 80 65,clip]{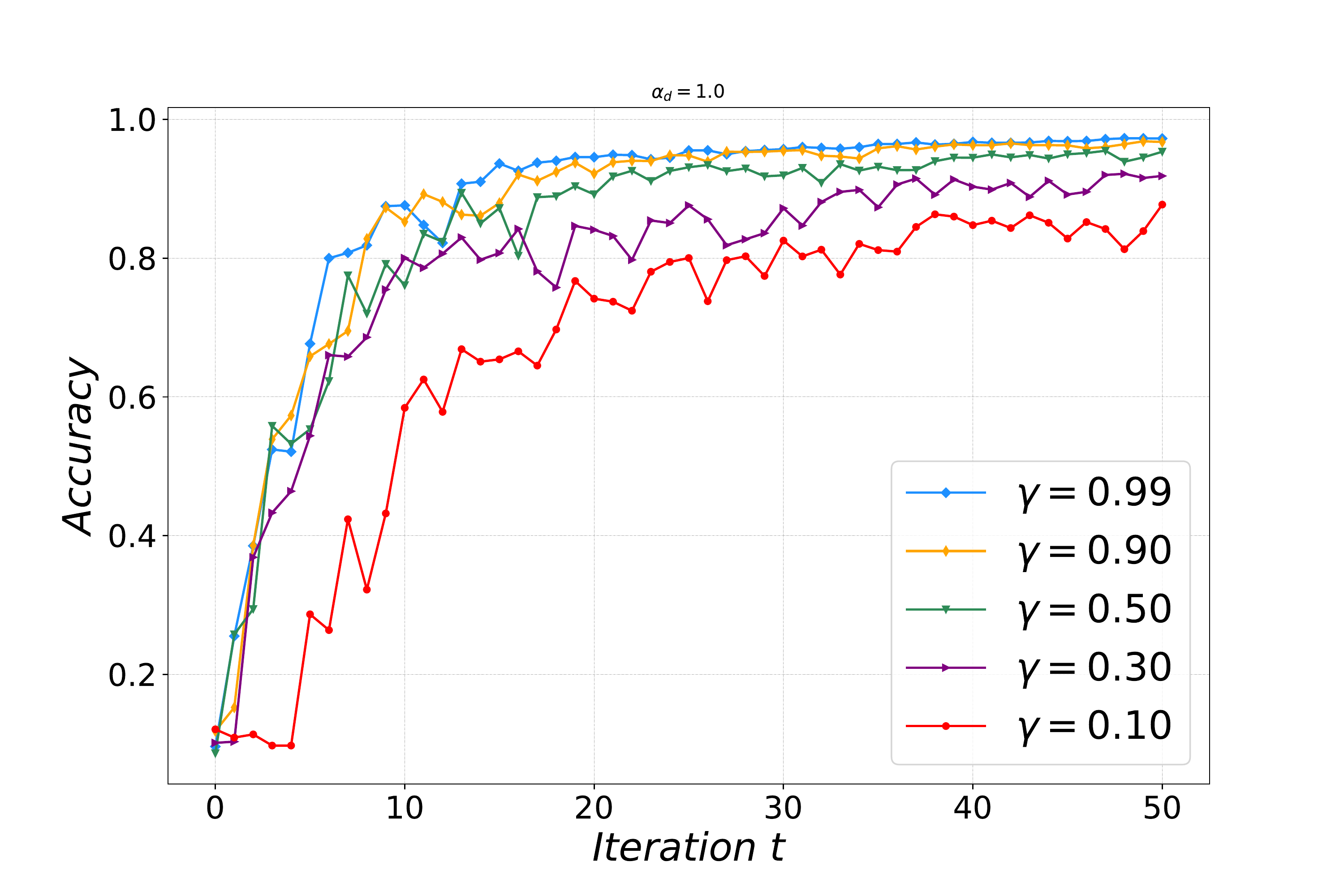}}
\subfloat[Loss Curve]{
    \label{MNIST_gamma_loss}
    \includegraphics[width=0.5\textwidth,  trim=40 10 80 65,clip]{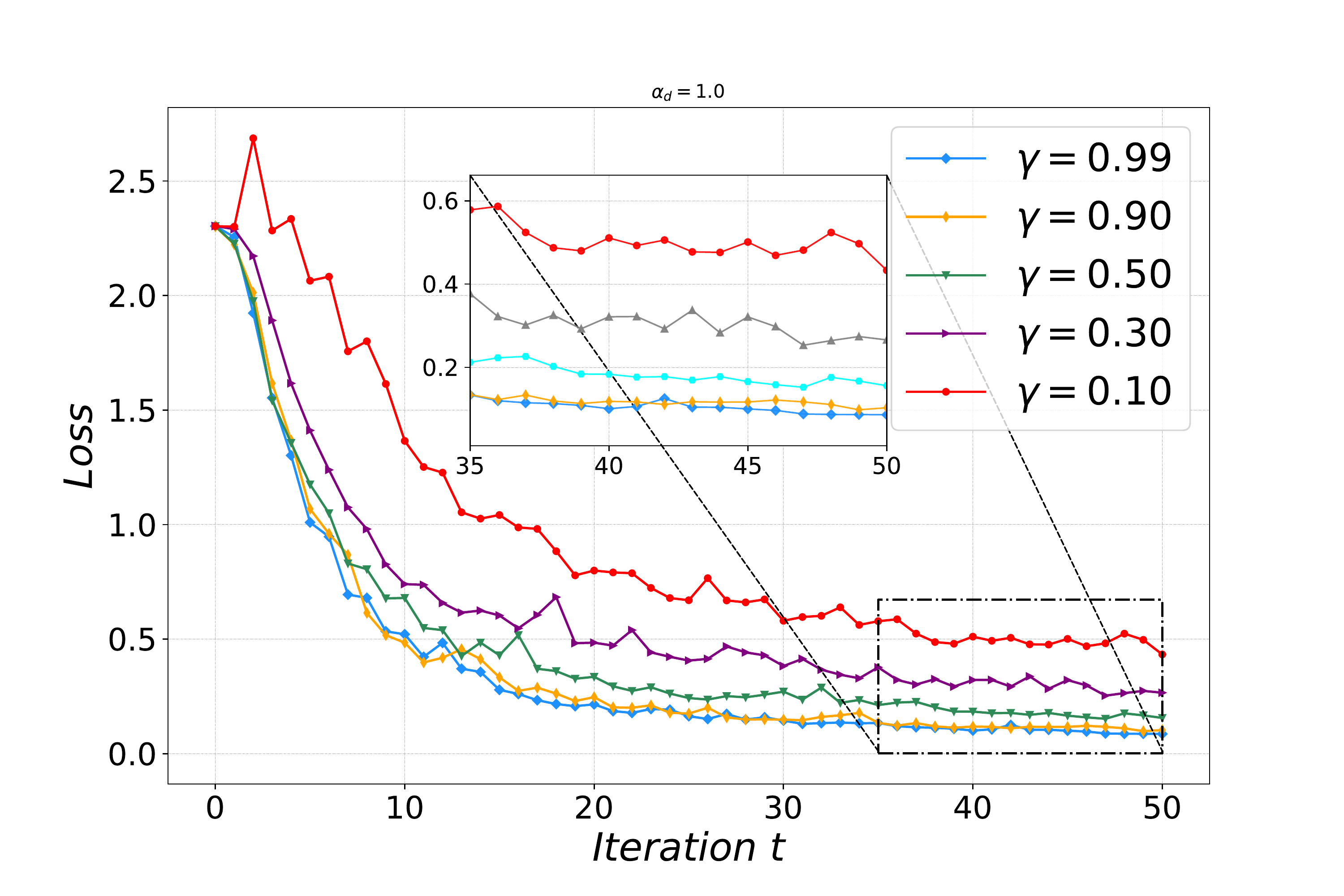}} 
\caption{The impact of the decay weight parameter $\gamma$ on the performance of FedEnt in the MNIST dataset.}
\label{ex_gamma}
\end{minipage}
\vspace{-18pt}
\end{figure}

\subsubsection{Effectiveness of hyperparameters}
To facilitate the practical application and promotion of the FedEnt, we have conducted a study on the impact of various hyperparameters on its model accuracy and convergence performance. First, we assess the impact of the aggregation weight $\beta$ delineated in Eq. (\ref{objective}) on the performance of the global model within the context of the MNIST dataset, featuring a Non-IID distribution. As depicted in Fig. \ref{beta}, an increase in $\beta$'s value directly enhances the global model's accuracy, a hastened convergence pace, and a loss reduction. This empirical evidence aligns with the theoretical exposition presented in Section \ref{sec_conv_an}, asserting that a heightened $\beta$ value signifies a more pronounced emphasis on minimizing the discrepancy among the local parameters of all clients, thereby augmenting model efficacy.  Subsequently, we delved into the impact of the learning rate decay parameter $\gamma$ in Eq. (\ref{eta_deacy}) on the global model performance. As shown in Fig. \ref{ex_gamma}(a), with the increase of $\gamma$, the global model accuracy of the FedEnt improves, and the convergence speed also accelerated notably. Simultaneously, we also noticed that the model loss gradually decreased as $\gamma$ increased, as depicted in Fig. \ref{ex_gamma}(b). These experimental results indicate that the value of $\gamma$ should be set larger (e.g., $\gamma=0.99$) to avoid the parameter fluctuations and achieve a faster convergence rate.

\section{Conclusion}\label{sec_conclusion}
In this paper, we discuss adaptive federated optimization with heterogeneous clients with Non-IID data distributions. Firstly, we utilize entropy theory to assess the degree of system disorder, based on which the adaptive learning rate for each client is designed to achieve a fast model convergence rate. Note that there is no communication among heterogeneous clients during local training epochs, we introduce the mean-field terms to estimate the components associated with other clients' local information over time. Through rigorous theoretical analysis, we demonstrate the existence and determination of the mean-field estimators and derive the closed-form of the decentralized adaptive learning rate for each client. In addition, the convergence property of our proposed adaptive FL algorithm FedEnt is analyzed. Finally, the extensive experimental results on different real-world datasets demonstrate the effectiveness of FedEnt, by achieving a faster convergence rate under the premise of keeping high accuracy compared with other state-of-the-art FL algorithms.

\bibliographystyle{IEEEtran}
\bibliography{ref}

\end{document}